\documentclass[a4paper,10pt]{article}

\usepackage{amsmath}
\usepackage{amssymb}
\usepackage{amsthm}
\usepackage{amsfonts}
\usepackage{amstext}
\usepackage{amsopn}
\usepackage{amsxtra}
\usepackage{graphicx}
\usepackage[latin1]{inputenc}
\usepackage{mathrsfs}
\usepackage{dsfont}
\newcommand\R{{\ensuremath {\mathbb R} }}
\newcommand\C{{\ensuremath {\mathbb C} }}

\newcommand\Z{{\ensuremath {\mathbb Z} }}

\newcommand\1{{\ensuremath {\mathds 1} }}

\renewcommand\phi{\varphi}

\newcommand{\alp}{\boldsymbol{\alpha}}

\newcommand{\gH}{\mathfrak{H}}
\newcommand{\gS}{\mathfrak{S}}

\newcommand{\wto}{\rightharpoonup}

\newcommand{\cF}{\mathcal{F}}

\newcommand\ii{{\ensuremath {\infty}}}
\newcommand\pscal[1]{{\ensuremath{\left\langle #1 \right\rangle}}}
\newcommand{\norm}[1]{ \left| \! \left| #1 \right| \! \right| }

\def\tr{\mathop{\rm tr}\nolimits} 
\newcommand{\E}{\mathcal{E}}
\newcommand{\Er}{\mathcal{E}_{\rm r}}

\newcommand{\boB}{B}
\newcommand{\cC}{\mathcal{C}}
\newcommand{\cQ}{\mathcal{Q}}
\newcommand{\cK}{\mathcal{K}}

\newcommand{\citeall}{\cite{HLS1}---\cite{HLSS}}
\newcommand{\citeallb}{\cite{HLS1}---\cite{HLSo}}

\newtheorem{thm}{Theorem}
\newtheorem{lemma}{Lemma}
\newtheorem{corollary}{Corollary}
\newtheorem{prop}[lemma]{Proposition}

\newtheorem{remark}{Remark}

\begin{document}
\title{Ground State and Charge Renormalization in a Nonlinear Model of Relativistic Atoms}

\author{Philippe GRAVEJAT, Mathieu LEWIN and \'Eric S\'ER\'E} 

\begin{center}
 \bf \Large Ground State and Charge Renormalization in a Nonlinear Model of Relativistic Atoms
\end{center}

\medskip

\begin{center}
 \large Philippe GRAVEJAT$^a$, Mathieu LEWIN$^b$ and \'Eric S\'ER\'E$^a$
\end{center}

\medskip

\begin{center}
\small

$^a$CEREMADE, UMR 7534, Universit\'e Paris-Dauphine, Place du Mar\'echal de Lattre de Tassigny, 75775 Paris
  Cedex 16, FRANCE.

\texttt{gravejat,sere@ceremade.dauphine.fr}

\medskip

$^b$CNRS \& Laboratoire de Mathématiques UMR 8088, Université de Cergy-Pontoise, 2 avenue Adolphe Chauvin, 95302 Cergy-Pontoise Cedex, FRANCE. 

\texttt{Mathieu.Lewin@math.cnrs.fr}
\end{center}

\begin{center}
 \it \today
\end{center}

\medskip

\begin{abstract}
We study the reduced Bogoliubov-Dirac-Fock (BDF) energy which allows to describe relativistic electrons interacting with the Dirac sea, in an external electrostatic potential. The model can be seen as a mean-field approximation of Quantum Electrodynamics (QED) where photons and the so-called exchange term are neglected. A state of the system is described by its one-body density matrix, an infinite rank self-adjoint operator which is a compact perturbation of the negative spectral projector of the free Dirac operator (the Dirac sea).

We study the minimization of the reduced BDF energy under a charge constraint. We prove the existence of minimizers for a large range of values of the charge, and any positive value of the coupling constant $\alpha$. Our result covers neutral and positively charged molecules, provided that the positive charge is not large enough to create electron-positron pairs. We also prove that the density of any minimizer is an $L^1$ function and compute the effective charge of the system, recovering the usual renormalization of charge: the physical coupling constant is related to $\alpha$ by the formula $\alpha_{\rm phys}\simeq \alpha(1+2\alpha/(3\pi)\log\Lambda)^{-1}$, where $\Lambda$ is the ultraviolet cut-off.
We eventually prove an estimate on the highest number of electrons which can be bound by a nucleus of charge $Z$. In the nonrelativistic limit, we obtain that this number is $\leq 2Z$, recovering a result of Lieb.

This work is based on a series of papers by Hainzl, Lewin, Séré and Solovej on the mean-field approximation of no-photon QED.
\end{abstract}

\bigskip


\section{Introduction}
In this paper, we study a model of Quantum Electrodynamics (QED) allowing to describe the behavior of relativistic electrons in an external field and interacting with the virtual electrons of the Dirac sea, in a mean-field type theory. This work should be seen as the continuation of previous papers by Hainzl, Lewin, Séré and Solovej \citeall, in which a more complicated model called Bogoliubov-Dirac-Fock (BDF) is considered. This project was mainly inspired of an important physical paper by Chaix and Iracane \cite{CI,Chaix} in which a model of the same kind was first proposed. 
We start by summarizing the physical motivation before defining the model properly.

Dirac introduced his operator in 1928 \cite{Dirac-28} with the purpose to describe the behavior of relativistic electrons. It is defined as
\begin{equation}
D^0=-i\sum_{k=1}^3\alpha_k\partial_k+\beta :=-i\alp\cdot \nabla+\beta
\label{dirac_free}
\end{equation}
where $\alp=(\alpha_1,\alpha_2,\alpha_3)$ and $\beta$ are the $4\times4$ Dirac matrices \cite{Thaller}. The operator $D^0$ acts on $L^2(\R^3,\C^4)$. Contrary to the non-relativistic Hamiltonian $-\Delta/2$, the operator $D^0$ is unbounded from below: $\sigma(D^0)=(-\ii,-1]\cup[1,\ii)$. This property is known to be the basic explanation of various peculiar physical phenomena like the possible creation of electron-positron pairs or the polarization of the vacuum. The model that we shall study is a rough approximation of Quantum Electrodynamics but it is able to reproduce many of these physical phenomena. We refer to \citeall~for more details.

In QED, one can write a formal Hamiltonian acting on the usual fermionic Fock space, in Coulomb gauge and neglecting photons \cite[Eq. (1)]{HLSo}. The mean-field approximation then consists in restricting formally this Hamiltonian to a special subclass of states in the Fock space, called the Hartree-Fock states. Any of these states is uniquely determined by its \emph{one-body density matrix} which is a self-adjoint operator $0\leq P\leq 1$ acting on $L^2(\R^3,\C^4)$. Often $P$ is an orthogonal projector. The QED energy then becomes a nonlinear functional in the variable $P$, which can be \textit{formally} written as follows
\begin{multline}
\E_{\rm QED}^\nu(P)=\tr(D^0(P-1/2))-\alpha \iint_{\R^3\times\R^3}\frac{\nu(x)\rho_{[P-1/2]}(y)}{|x-y|}dx\,dy \\
+\frac{\alpha}{2}\iint_{\R^3\times\R^3}\frac{\rho_{[P-1/2]}(x)\rho_{[P-1/2]}(y)}{|x-y|}dx\,dy
- \frac{\alpha}{2}\iint_{\R^3\times\R^3}\frac{|(P-1/2)(x,y)|^2}{|x-y|}dx\,dy,
\label{formal_QED_energy}
\end{multline}
where for any operator $Q$ acting on $L^2(\R^3,\C^4)$ with kernel $Q(x,y)$, $\rho_Q$ is formally defined as $\rho_Q(x)=\tr_{\C^4}(Q(x,x))$. Recall $Q(x,y)$ acts on $4$-spinors, i.e. is a $4\times4$ complex hermitian matrix. The first term of \eqref{formal_QED_energy} is the kinetic energy of the particles, whereas the second term describes the interaction with an external electrostatic field created by a smooth distribution of charge $\nu$ (describing for instance a system of classical nuclei). The last two terms account for the interaction between the particles themselves. We have chosen a system of units such that $\hbar=c=1$, and also such 
that the mass $m_e$ of the electron is normalized to 1. The constant $\alpha=e^2$ (where $e$ is the bare charge of an electron) is a small number called the \emph{Sommerfeld fine-structure constant}.

Expression \eqref{formal_QED_energy} is purely formal: when $P$ is an orthogonal projector on $L^2(\R^3,\C^4)$, $P-1/2$ is never compact and none of the terms above makes sense \emph{a priori}. However, it is possible to give a meaning to \eqref{formal_QED_energy} by restricting the system to a box and imposing an ultraviolet cut-off. One can then study the thermodynamic limit, i.e. the behavior of the energy and of the minimizers when the size of the box goes to infinity (but the ultraviolet cut-off is fixed). This approach was the main purpose of \cite{HLSo}.

The last two terms of \eqref{formal_QED_energy} are respectively called the \emph{direct term} and the \emph{exchange term}. In theoretical studies of the Hartree-Fock model, the exchange term is sometimes neglected \cite{Solovej2}. The above energy then becomes (formally) convex, a very interesting simplification both from a theoretical and numerical point of view. Refined models exist: in relativistic density functional theory for instance, the exchange term is approximated by a function of the density $\rho_{[P-1/2]}$ and its derivatives only, see, e.g., the review \cite{Engel}. Neglecting the last term, one is led to consider the following \emph{reduced} formal functional
\begin{multline}
\E_{\text{r-QED}}^\nu(P)=\tr(D^0(P-1/2))-\alpha \iint_{\R^3\times\R^3}\frac{\nu(x)\rho_{[P-1/2]}(y)}{|x-y|}dx\,dy \\
+\frac{\alpha}{2}\iint_{\R^3\times\R^3}\frac{\rho_{[P-1/2]}(x)\rho_{[P-1/2]}(y)}{|x-y|}dx\,dy.
\label{formal_red_QED_energy}
\end{multline}

As usual, one is interested in finding states having lowest energy, possibly  in a specific subclass. In QED, a global minimizer in the Fock space is interpreted as being the vacuum, whereas other states (containing a finite number $q$ of real electrons for example) are obtained by assuming a charge constraint. 
When the external field vanishes ($\nu\equiv0$) and for any values of the coupling constant $\alpha\geq0$, one easily proves that $\E_{\text{r-QED}}^0$ has a unique global minimizer which is the negative spectral projector of the free Dirac operator:
$$P^0_-:=\chi_{(-\ii,0]}(D^0).$$
The precise mathematical statement is that when the system is restricted to a box of size $L$ with an ultraviolet cut-off $\Lambda$, the above energy is well-defined; it has a unique minimizer 
$$P_L=\chi_{(-\ii,0]}(D^0_L)$$
where $D^0_L$ is the Dirac operator acting on the box with periodic boundary conditions. The sequence $P^0_L$ converges (in a weak sense) to $P^0_-$ which is thus interpreted as the unique global minimizer of $P\mapsto\E_{\text{r-QED}}^0(P)$. If the exchange term is not neglected, the situation is more complicated and we refer to \cite{HLSo} where the thermodynamic limit was carried out. 

The fact that $P^0_-$ is found to be the global minimizer of our formal energy is not physically surprising. This corresponds to the usual Dirac picture \cite{Dirac-28,Dirac-30,Dirac-34a,Dirac-34b} which consists in assuming that the vacuum should be seen as an infinite system of virtual particles occupying all the negative energy states of the free Dirac operator. Notice however that when the exchange term is taken into account, this picture is no longer valid: $P^0_-$ does not describe the free vacuum which is instead solution of a complicated translation-invariant nonlinear equation, see \cite{HLSo}.

We want to emphasize the importance of the subtraction of half the identity in all the terms of the above energy \eqref{formal_red_QED_energy}. Indeed, the kernel of the translation-invariant operator $P^0_--1/2$ is
$$(P^0_--1/2)(x,y)=(2\pi)^{-3/2}f(x-y)\ \text{where}\ \hat{f}(k)=-\frac{D^0(k)}{2|D^0(k)|}.$$
If we assume that there is a cut-off $\Lambda$ in the Fourier domain, i.e. ${\rm supp}(\hat{f})\subseteq B(0,\Lambda)$, it is then possible to compute the density
\begin{equation}
\rho_{[P^0_--1/2]}=(2\pi)^{-3/2}\tr_{\C^4}(f(0))=(2\pi)^{-3}\int_{B(0,\Lambda)}\tr_{\C^4}(\hat{f}(k))dk\equiv0,
\label{density_vanishes}
\end{equation}
the Dirac matrices being trace-less. We therefore obtain that the free vacuum has no density of charge, which is comforting physically.

When the external field does not vanish, the main idea is then to subtract the (infinite) energy of the free vacuum $\E_{\text{r-QED}}^0(P^0_-)$ to \eqref{formal_red_QED_energy}, in order to obtain a finite quantity. This yields the so-called (formal) \emph{reduced-Bogoliubov-Dirac-Fock energy} (rBDF) which was already studied in \cite{HLS2} and is more easily expressed in terms of the difference $Q=P-P^0_-$,
\begin{eqnarray}
\E_{\rm r}^\nu(P-P^0_-) & = & ``\E_{\text{r-QED}}^\nu(P)-\E_{\text{r-QED}}^0(P^0_-)"\nonumber\\
 & = & \tr(D^0(P-P^0_-))-\alpha \iint_{\R^6}\frac{\nu(x)\rho_{[P-P^0_-]}(y)}{|x-y|}dx\,dy\nonumber \\
 & &  \qquad\qquad +\frac{\alpha}{2}\iint_{\R^6}\frac{\rho_{[P-P^0_-]}(x)\rho_{[P-P^0_-]}(y)}{|x-y|}dx\,dy.
\label{red-BDF_formal}
\end{eqnarray}
Note that we have used \eqref{density_vanishes}. What we have gained is that $Q=P-P^0_-$ can now be a compact operator (it will indeed be Hilbert-Schmidt). We recall that $P$ is the density matrix of our Hartree-Fock state, hence it satisfies $0\leq P\leq 1$ which translates on $Q$ as $-P^0_-\leq Q\leq 1-P^0_-:=P^0_+$. 

A (formal) global minimizer $Q$ of $\Er^\nu$ is interpreted as the polarized vacuum in the presence of the external density $\nu$. Formally, it solves the self-consistent equation
\begin{equation}
 \left\{\begin{array}{l}
Q=\chi_{(-\ii,0)}(D_Q)-P^0_-\\
D_Q=D^0+\alpha(\rho_Q-\nu)\ast|\cdot|^{-1}.
\end{array}\right.
\label{SCF_eq_vacuum_intro}
\end{equation}
In order to describe a physical system containing a finite number $q$ of real electrons, it is necessary to minimize the above energy not on the full class of states, but rather in a chosen charge sector, i.e. over states satisfying the formal charge constraint $\text{``}\tr(Q)=\tr(P-P^0_-)=q\text{''}$. Then a minimizer will satisfy the following equation
\begin{equation}
\left\{\begin{array}{l}
Q=\chi_{(-\ii,\mu)}(D_Q)-P^0_-+\delta\\
D_Q=D^0+\alpha(\rho_Q-\nu)\ast|\cdot|^{-1}
\end{array}\right.
\label{SCF_eq_mol_intro}
\end{equation}
where $\mu$ is a Lagrange multiplier due to the charge constraint and interpreted as a chemical potential. The operator $\delta$ is a finite rank operator satisfying $0\leq \delta\leq 1$ and ${\rm Ran}(\delta)\subset\ker(D_Q-\mu)$. Notice the number $q$ does not need to be an integer as one may want to describe mixed states (in which case $\delta\neq0$).

We see that in both cases (minimization with or without a charge constraint), a minimizer always corresponds to filling energies of an effective Dirac operator up to some Fermi level $\mu$. This corresponds to original ideas of Dirac. For the general BDF theory, the idea that one can have a bounded below functional whose minimizer satisfies this kind of equation was first proposed by Chaix and Iracane \cite{CI,Chaix}.

In this paper, we shall prove that the range of $q$'s such that minimizers exist is an interval $[q_m,q_M]\subset\R$ which contains both the charge of the polarized vacuum (the global minimizer of the energy, solution of \eqref{SCF_eq_vacuum_intro}) denoted by $q_0$, and $Z=\int_{\R^3}\nu$. This proves the existence of neutral molecules and of positively charged molecules the charge of which is not too big, because in this case one has $q_0=0$. This extends previous results proved for the BDF theory with the exchange term in \cite{HLS3}: sufficient conditions were given for the existence of minimizers, but these conditions could only be checked in the nonrelativistic or the weak coupling limits. In the present paper, we shall also give interesting properties of a minimizer when it exists, and provide a bound on the maximal number of electrons which can be bound by a nucleus of charge $Z$, following ideas of Lieb \cite{Lieb}.

\medskip

The mathematical formulation and the proofs of the above statements are not straightforward.

The first (and main) difficulty is that we do not expect that a solution $Q$ of Equations \eqref{SCF_eq_vacuum_intro} or \eqref{SCF_eq_mol_intro} is a trace-class operator. Indeed our results below will imply that in most cases it \emph{cannot be trace-class}. This is a big problem as in the energy \eqref{formal_red_QED_energy} the first term is expressed as a trace, as well as the total charge of the system which we formally wrote ``$\tr(Q)$'' in the previous paragraphs. This issue was solved in \cite{HLS1} where it was proposed to generalize the trace functional and to define the trace counted relatively to the free vacuum $P^0_-$ as
$$\tr_{P^0_-}(Q):=\tr(P^0_+QP^0_+)+\tr(P^0_-QP^0_-).$$
As we shall see, any minimizer $Q$ will have a finite so-defined $P^0_-$-trace, which does not mean that $Q$ is trace-class.

If we do not expect $Q$ to be trace-class, there is a problem in defining the density of charge $\rho_Q$. Indeed it is known that in QED there are several divergences which need to be removed by means of an ultraviolet cut-off. In previous works \citeall, a sharp cut-off $\Lambda$ was imposed: the space $L^2(\R^3,\C^4)$ was replaced by its subspace consisting of functions that have a Fourier transform with support in the ball of radius $\Lambda$. This allowed to give a solid mathematical meaning to the energy \eqref{red-BDF_formal}. In \cite{HLS1,HLS2}, it was proved that the energy has a global minimizer $Q$, solution of \eqref{SCF_eq_vacuum_intro}. In \cite{HLS3}, sufficient conditions were given on $q$ to ensure the existence of a ground state in the charge sector $q$ with the exchange term. They could only be checked in the nonrelativistic or the weak coupling limit.

In this paper, we propose other kinds of cut-offs which seem better for obtaining decay properties of the density of charge\footnote{A similar remark was made in \cite{LieLo} in the context of non-relativistic QED.}. Essentially, they consist in replacing the Dirac operator $D^0$ by $D^\zeta(p)=(\alp\cdot p+\beta)(1+\zeta(|p|^2/\Lambda^2))$ where $\zeta$ is a smooth function growing fast enough at infinity. We call these cut-offs \emph{smooth} in contrast to the previous sharp cut-off. But many of our results will also be valid in the sharp cut-off case.

Even with an ultraviolet cut-off, a minimizer $Q$ will in general not be trace-class. But we shall be able to prove that anyway its density of charge is an $L^1$ function: $\rho_Q\in L^1(\R^3)$. This information can then be used to prove the existence of all atoms and molecules which are either neutral or positively charged and do not have a too strong positive nuclear density.
Also we shall prove a formula which relates the integral of $\rho_Q$ and $q=\tr_{P^0_-}(Q)$ of the form
\begin{equation}
 \int_{\R^3}\rho_Q -Z\simeq\frac{q-Z}{1+2/(3\pi)\alpha\log\Lambda}\label{formula_int_rho_intro}
\end{equation}
(see Theorem \ref{thm_L1} for a precise statement depending on the chosen cut-off $\Lambda$). When $q\neq Z$, this proves that $\int_{\R^3}\rho_Q\neq q=\tr_{P^0_-}(Q)$, hence $Q$ cannot be trace-class.

The fact that a minimizer is not trace-class but its density is anyway an $L^1$ function can first be thought of as being a technical issue. But Equation \eqref{formula_int_rho_intro} has a relevant physical interpretation. It means that the total observed charge $\int_{\R^3}\rho_Q -Z$ is different from the real charge $q-Z$ of the system. Hence the mathematical property that a minimizer is not trace-class is well interpreted physically in terms of \emph{charge renormalization}. We even recover a standard charge renormalization formula in QED, see \cite[Eq. $(8)$]{Lan84} and \cite[Eq. $(7.18)$]{IZ}, although we use a simple model without photons and within the Hartree-Fock approximation with the exchange term removed.

As announced before, we shall prove in this paper that minimizers exist if and only if $q\in[q_m,q_M]$, an interval which contains both $Z$ and the charge $q_0$ of the polarized vacuum. We shall also derive some bounds on $q_m$ and $q_M$, assuming that the nuclear charge distribution is not too strong. Essentially we prove that $q_m<0$ is very small and that 
$$Z\leq q_M\leq 2Z+\underset{\alpha\to0}{o}(1).$$
 In the nonrelativistic limit we recover the usual bound of the reduced Hartree-Fock model which can be obtained by a method of Lieb \cite{Lieb}.

\medskip

In the next section, we define the reduced BDF energy \eqref{red-BDF_formal} properly and state our main results. Proofs are given in Section \ref{sec_proof}.

\bigskip

\noindent\textbf{Acknowledgment.} M.L. and E.S. acknowledge support from the ANR project ``ACCQUAREL'' of the French ministry of research.

\section{Model and main results}
In the whole paper, we denote by $\gS_p(\gH)$ the usual Schatten class of operators $Q$ acting on a Hilbert space $\gH$ and such that $\tr(|Q|^p)<\ii$. We use the notation $Q^{\epsilon\epsilon'}:=P^0_\epsilon QP^0_{\epsilon'}$ for any $\epsilon,\epsilon'\in\{\pm\}$.
A self-adjoint operator $Q$ acting on $\gH$ is said to be $P^0_-$-trace class \cite{HLS1} if $Q\in\gS_2(\gH)$ and $Q^{++},\ Q^{--}\in\gS_1(\gH)$. We then define its $P^0_-$-trace as
$$\tr_{P^0_-}(Q)=\tr(Q^{--})+\tr(Q^{++}).$$
The space of $P^0_-$-trace class operators on $\gH$ will be denoted by $\gS_1^{P^0_-}(\gH)$. We refer to \cite{HLS1} where important properties of this generalization of the trace functional are provided.

\subsection{Ultraviolet regularization}
It is well-known that in Quantum Electrodynamics a cut-off is mandatory \cite{BD,IZ}. There are two sources of divergence in the Bogoliubov-Dirac-Fock model. The first is the negative continuous spectrum of the Dirac operator, which is cured by the subtraction of the (infinite) energy of the Dirac sea, as explained above. The second source of divergence is the rather slow growth of the Dirac operator for large momenta: $D^0$ only behaves linearly in $p$ at infinity\footnote{Notice a model similar to the reduced-BDF theory was recently studied for non-relativistic crystals in the presence of defects \cite{CDL}, in which case a cut-off is not necessary because of the presence of the Laplacian instead of $D^0$.}.

This can be cured by imposing a sharp cut-off on the space, i.e. by replacing $L^2(\R^3,\C^4)$ by its subspace
\begin{equation}
\gH_\Lambda:=\left\{f\in L^2(\R^3,\C^4)\ |\ \text{supp}(\widehat{f})\subseteq B(0,\Lambda)\right\}.
\label{def_H_Lambda}
\end{equation}
Notice $D^0\gH_\Lambda\subset\gH_\Lambda$.
This simple approach was chosen in previous works \citeall.

However, when looking at decay properties of the electronic density, it might be more adapted to instead increase the growth of the Dirac operator at infinity. This means we replace $D^0$ by the operator
\begin{equation}
 D^\zeta(p):=(\alp\cdot p+\beta)\left(1+\zeta\left(\frac{|p|^2}{\Lambda^2}\right)\right)
\label{def_D_cut_off}
\end{equation}
where $\zeta:[0,\ii)\mapsto[0,\ii)$ grows fast enough at infinity. The operator $D^\zeta$ is self-adjoint on $\gH=L^2(\R^3,\C^4)$ with domain
$${\cal D}(D^\zeta):=\left\{f\in L^2(\R^3,\C^4)\ |\ \left(1+|p|\zeta\left(\frac{|p|^2}{\Lambda^2}\right)\right)^{1/2}\widehat{f}(p)\in L^2(\R^3,\C^4)\right\}.$$

We remark that the case of the sharp cut-off \eqref{def_H_Lambda} formally corresponds to 
\begin{equation}
 \zeta(x)=\left\{\begin{array}{ll}
0 & \text{if } |x|\leq1;\\
+\ii & \text{otherwise.}
\end{array}\right.
\label{zeta_infini}
\end{equation}

\medskip

In this work, we shall consider both cases \eqref{def_H_Lambda} and \eqref{def_D_cut_off}. We assume throughout the whole paper that 
\begin{description}
 \item[$\bullet$ either] $\gH=\gH_\Lambda$ and $\zeta\equiv0$ (or equivalently $\zeta$ given by \eqref{zeta_infini});
\item[$\bullet$ or] $\gH=L^2(\R^3,\C^4)$ and $\zeta$ satisfies the following properties:
\end{description}

\vspace{-0,4cm}
\begin{equation}
\zeta\in C^3([0,\ii)) \text{ is non-decreasing and }\ \zeta(0)=0,
\label{prop_zeta_1}
\end{equation}
\begin{equation}
\zeta(x)\geq \varepsilon x^{\varepsilon/2}\1(x\geq1)\quad \text{for some } \varepsilon>0, 
\label{prop_zeta_2}
\end{equation}
\begin{equation}
(1+|x|^p)\left|\zeta^{(p)}(x)\right|\leq C(1+\zeta(x))\quad \text{for } p=1,2,3.
\label{prop_zeta_3}
\end{equation}
Many of our results will be true under weaker assumptions on $\zeta$ but we shall restrict ourselves to \eqref{prop_zeta_1}--\eqref{prop_zeta_3} for simplicity.
We notice that under these assumptions, the spectrum of $D^\zeta$ is the same as the one of $D^0$:
$$\sigma(D^\zeta)=(-\ii;-1]\cup[1;\ii).$$
Also the negative spectral projector of $D^\zeta$ is the same as the one of $D^0$:
$$P^0_-=\chi_{(-\ii,0]}(D^0)=\chi_{(-\ii,0]}(D^\zeta).$$

In the whole paper, we shall consider perturbations of $D^\zeta$ of the form $D^\zeta+\rho\ast|\cdot|^{-1}$ where $\rho$ belongs to the so-called Coulomb space 
\begin{equation}
 \cC:=\{\rho \in \mathcal{ S}'(\R^3) \ | \  D(\rho,\rho)<\ii\}
\label{def_cC}
\end{equation}
where 
\begin{equation}
 D(f,g)=4\pi\int_{\R^3}|k|^{-2}\overline{\widehat{f}(k)}\widehat{g}(k)dk.
\label{def_Coulomb}
\end{equation}
Notice the dual space of $\cC$ is the Beppo-Levi space
$$\cC':= \left\{ V \in  L^6(\R^3) \ | \ \nabla V \in L^2(\R^3) \right\}.$$
\begin{lemma}\label{lem_weyl} We assume that $\gH=\gH_\Lambda$ and $\zeta=0$, or that $\gH=L^2(\R^3,\C^4)$ and $\zeta$ satisfies \eqref{prop_zeta_1}--\eqref{prop_zeta_3}.
 For any $\rho\in\cC$, the operator $D^\zeta+\rho\ast|\cdot|^{-1}$ defined on the same domain as $D^\zeta$ is self-adjoint and satisfies:
$$\sigma_{\rm ess}(D^\zeta+\rho\ast|\cdot|^{-1})=\sigma_{\rm ess}(D^\zeta)=(-\ii,-1]\cup[1,\ii).$$
\end{lemma}
\begin{proof}
We denote $V:=\rho\ast|\cdot|^{-1}$. We have $V|D^\zeta|^{-1}$ is in $\gS_6(\gH)$, hence is compact. This is because we can use the Kato-Seiler-Simon inequality (see \cite{SeSi} and \cite[Thm 4.1]{Simon})
\begin{equation}
 \forall p\geq2,\qquad \norm{f(-i\nabla)g(x)}_{\gS_p}\leq (2\pi)^{-3/p}
\norm{g}_{L^p(\R^3)}\norm{f}_{L^p(\R^3)}
\label{KSS}
\end{equation} 
 and obtain
\begin{equation}
\norm{V|D^\zeta|^{-1}}_{\gS_6(\gH)}\leq C\norm{V}_{L^6(\R^3)}\norm{|D^\zeta(\cdot)|^{-1}}_{L^6(\R^3)}\leq C\norm{\nabla V}_{L^2}=C\norm{\rho}_\cC. 
\label{estim_weyl}
\end{equation}
Lemma \ref{lem_weyl} is then an application of a criterion by Weyl \cite[Sec. XIII.4]{RS4}. 
\end{proof}

\subsection{Definition of the reduced-BDF energy}
We recall that $\gH=L^2(\R^3,\C^4)$ or $\gH=\gH_\Lambda$ depending on the chosen cut-off. We need to provide a correct setting for the rBDF energy. When $\gH=\gH_\Lambda$, this was done in \citeall. When $\gH=L^2(\R^3,\C^4)$, this is done similarly to the crystal case studied in \cite{CDL}. 
We introduce the following Banach space:
\begin{multline}
\cQ:=\bigg\{Q\in\gS_2(\gH)\quad |\quad Q^*=Q,\ |D^\zeta|^{1/2}Q\in \gS_2(\gH),\\  |D^\zeta|^{1/2}Q^{++}|D^\zeta|^{1/2}\in\gS_1(\gH),\ |D^\zeta|^{1/2}Q^{--}|D^\zeta|^{1/2}\in\gS_1(\gH)\bigg\}
\label{def_Q}
\end{multline}
with associated norm
\begin{multline}
\norm{Q}_\cQ:=\norm{|D^\zeta|^{1/2}Q}_{\gS_2(\gH)}+\norm{|D^\zeta|^{1/2}Q^{++}|D^\zeta|^{1/2}}_{\gS_1(\gH)}\\+\norm{|D^\zeta|^{1/2}Q^{--}|D^\zeta|^{1/2}}_{\gS_1(\gH)}.
\label{norm_Q}
\end{multline}
We notice that when $\gH=\gH_\Lambda$ and $\zeta=0$, one has $\cQ=\gS^{P^0_-}_1(\gH_\Lambda)$ as chosen in \citeallb. In the general case, we only have $\cQ\subset \gS^{P^0_-}_1(\gH)$.
We recall that $\gS_1(\gH)$ is the dual of the space of compact operators acting on $\gH$. Hence $\gS_1(\gH)$ can be endowed with the associated weak-$\ast$ topology where $A_n\wto A$ in  $\gS_1(\gH)$ means that $\tr(A_nK)\to\tr(AK)$ for any compact operator $K$. Together with the fact that $\gS_2(\gH)$ is a Hilbert space, this defines a weak topology on $\cQ$.

We also introduce the following convex subset of $\cQ$:
\begin{equation}
 \cK:=\left\{Q\in\cQ\ |\ -P^0_-\leq Q\leq P^0_+\right\}
\label{def_K}
\end{equation}
which is the closed convex hull of states of the form $Q=P-P^0_-\in\cQ$ where $P$ is an orthogonal projector acting on $\gH$. It is clear that $\cK$ is closed both for the strong and the weak-$\ast$ topology of $\cQ$.
As we shall see, the reduced BDF energy will be coercive and weakly lower semi-continuous on $\cK$.

Besides, the number $\tr_{P^0_-}(Q)$ can be interpreted as the charge of the system measured with respect to that of the unperturbed Dirac sea $P^0_-$, see \citeall. Note that the constraint $-P^0_-\leq Q\leq P^0_+$ in \eqref{def_K} is indeed equivalent \cite{BBHS,HLS1} to the inequality
\begin{equation}
0\leq Q^2\leq Q^{++}-Q^{--}
\label{equivalent_condition} 
\end{equation}
and implies in particular that $Q^{++}\geq 0$ and $Q^{--}\leq0$ for any $Q\in\cK$.

We need to define the density $\rho_Q$ of any state $Q\in\cQ$. When $\gH=\gH_\Lambda$, this is easy as any $Q\in\cQ$ has a smooth kernel $Q(x,y)$ (this is because the Fourier transform $\widehat{Q}(p,q)\in L^2(B(0,\Lambda)^2)$). This property was used in \cite{HLS1}--\cite{HLSo} to properly define the density of charge.
In the case where $\gH=L^2(\R^3,\C^4)$ and $\zeta\neq0$, this is a bit more involved.
The following is similar to \cite[Lemma 1]{HLS3} and \cite[Prop. 1]{CDL} (we recall that $\cC$ was defined above in \eqref{def_cC}):

\begin{prop}[Definition of the density $\rho_Q$ for $Q\in\cQ$] \label{prop_def_rho} We assume that $\gH=\gH_\Lambda$ and $\zeta=0$, or that $\gH=L^2(\R^3,\C^4)$ and $\zeta$ satisfies \eqref{prop_zeta_1}--\eqref{prop_zeta_3}.

Let $Q\in\cQ$. Then $QV\in \gS_1^{P^0_-}(\gH)$ for any $V\in\cC'$. Moreover there exists a constant $C$ (independent of $Q$ and $V$) such that
$$|\tr_{P^0_-}(QV)|\leq C\norm{Q}_{\cQ}\norm{V}_{\cC'}.$$
Hence, there exists a continuous linear form $Q\in\cQ\mapsto\rho_Q\in\cC$ which satisfies
$$\tr_{P^0_-}(QV)=\;_{\cC'}\!\pscal{V,\rho_Q}\;\!_\cC$$
for any $V\in\cC'$ and any $Q\in \cQ$. 
Eventually when $Q\in\cQ\cap \gS_1(\gH)$, then $\rho_Q(x)=\tr_{\C^4}Q(x,x)$ where $Q(x,y)$ is the integral kernel of $Q$.
\end{prop}

The proof of Proposition \ref{prop_def_rho} is given in Section \ref{proof_prop_def_rho} below.

Let us now define the reduced Bogoliubov-Dirac-Fock (rBDF) energy. In the whole paper, we use the notation, for any $Q\in\cQ$,
\begin{equation}
\tr_{P^0_-}(D^\zeta Q):=\tr\left(|D^\zeta|^{1/2}(Q^{++}-Q^{--})|D^\zeta|^{1/2}\right).
\label{def_kinetic_energy}
\end{equation}
When $D^\zeta Q\in\gS_1^{P^0_-}(\gH)$, this coincides with the definition of the generalized trace introduced above. The rBDF energy reads:
\begin{equation}
\label{def_energy}
\boxed{\Er^\nu(Q)=\tr_{P^0_-}(D^\zeta Q)-\alpha D(\nu,\rho_Q)+\frac\alpha2 D(\rho_Q,\rho_Q)}
\end{equation}
where we recall that $D(\cdot,\cdot)$ was defined in \eqref{def_Coulomb}.
In \eqref{def_energy}, $\nu$ is an external density which will be assumed to belong to $L^1(\R^3)\cap \cC$.
We use the notation $\int_{\R^3}\nu=Z$.
The energy $\Er^\nu$ is well-defined \cite{HLS1,HLS3} on the convex set $\cK$. By \eqref{equivalent_condition}, we have 
\begin{align}
\tr_{P^0_-}(D^\zeta Q) &= \norm{|D^\zeta|^{1/2} Q^{++}|D^\zeta|^{1/2}}_{\gS_1(\gH)}+\norm{|D^\zeta|^{1/2} Q^{--}|D^\zeta|^{1/2}}_{\gS_1(\gH)}\nonumber\\
 & \geq \norm{|D^\zeta| Q}_{\gS_2(\gH)}^2.
\label{kinetic_energy_coercive} 
\end{align}
Together with
$$-\alpha D(\nu,\rho_Q)+\frac\alpha2 D(\rho_Q,\rho_Q)\geq -\frac\alpha2 D(\nu,\nu),$$
this proves both that $\Er^\nu$ is bounded from below on $\cK$,
$$\forall Q\in\cK,\qquad\Er^\nu(Q)\geq -\frac\alpha2 D(\nu,\nu),$$
and that it is coercive for the topology of $\cQ$.

Since $\Er^\nu$ is convex on $\cK$ and weakly lower semi-continuous, it has a global minimizer $\bar Q_{\rm vac}$, interpreted as the polarized vacuum in the presence of the external field created by the density $\nu$. This was remarked in \cite[Theorem 3]{HLS2}. Assuming that $\ker(D_{\bar Q_{\rm vac}})=\{0\}$ where 
$$D_{\bar Q_{\rm vac}}=D^\zeta+\alpha(\rho_{\bar Q_{\rm vac}}-\nu)\ast|\cdot|^{-1}$$
is the mean field operator, then one can adapt the proof of \cite[Theorem 3]{HLS2} to get that $\bar Q_{\rm vac}$ is unique and is a solution of the nonlinear equation $\bar Q_{\rm vac}=\chi_{(-\ii,0]}(D_{\bar Q_{\rm vac}})-P^0_-$. The charge of the polarized vacuum is $-eq_0$ where 
$$\boxed{q_0=\tr_{P^0_-}(\bar Q_{\rm vac}).}$$
When $\alpha D(\nu,\nu)^{1/2}$ is not too large \cite[Eq. (15)]{HLS2}, it was proved that $q_0=0$. However in general electron-positron pairs can appear, giving rise to a charged vacuum. When $\ker(D_{\bar Q_{\rm vac}})\neq\{0\}$, then $\E_{\rm r}^\nu$ does not have a unique global minimizer on $\cK$, but it will be proved that $q_0$ is anyway a uniquely defined quantity.

\subsection{Existence of minimizers with a charge constraint}
We are interested in the following minimization problem
\begin{equation}
\label{def_min}
\boxed{E_{\rm r}^\nu(q)=\inf_{Q\in\mathcal{Q}(q)}\Er^\nu(Q)}
\end{equation}
where the sector of charge $-eq$ is by definition
$$\cQ(q):=\{Q\in\mathcal{Q},\ \tr_{P^0_-}(Q)=q\}$$
and $q$ is any real number. Of course in Physics $q\in\Z$ but it is convenient to allow any real value. It will be proved below that $q\to E_{\rm r}^\nu(q)$ is a Lipschitz and convex function. Notice that if $\bar Q$ is a global minimizer of $\E_{\rm r}^\nu$ on $\cQ$, then $q_0=\tr_{P^0_-}(\bar Q)$ minimizes $q\to E_{\rm r}^\nu(q)$.

The existence of minimizers to \eqref{def_min} is not obvious: although $\Er^\nu$ is convex and weakly lower semi-continuous, and $\cQ(q)$ is itself a convex set, the linear form $Q\mapsto \tr_{P^0_-}(Q)$ is not weakly continuous. Hence $\cQ(q)$ is not closed for the weak topology.
Our main result is the following theorem, whose proof is given in Section \ref{proof_thm_exists} below.
\begin{thm}[Existence of atoms and molecules in the reduced BDF model]\label{exists} We assume that $\gH=\gH_\Lambda$ and $\zeta=0$, or that $\gH=L^2(\R^3,\C^4)$ and $\zeta$ satisfies \eqref{prop_zeta_1}--\eqref{prop_zeta_3}. Let be $\alpha\geq0$, $\nu\in L^1(\R^3)\cap \cC$ and denote $Z=\int_{\R^3}\nu\in\R$. Then there exists $q_m\in[-\ii,\ii)$ and $q_M\in[q_m,\ii]$ such that

\bigskip

\noindent $(i)$ $[q_m,q_M]$ is the largest interval on which $q\to E^\nu_{\rm r}(q)$ is strictly convex. If $q_M<\ii$, then $E^\nu_{\rm r}(q)=E^\nu_{\rm r}(q_M)+q-q_M$ for any $q>q_M$. If $q_M>-\ii$, then $E^\nu_{\rm r}(q)=E^\nu_{\rm r}(q_m)+q_m-q$ for any $q<q_m$;

\bigskip

\noindent $(ii)$ the interval $[q_m,q_M]$ contains both $Z$ and the unique minimizer $q_0$ of $q\to E^\nu_{\rm r}(q)$;

\bigskip

\noindent $(iii)$ if $q\notin [q_m,q_M]$, then $\Er^\nu$ has \emph{no minimizer in the charge sector} $\cQ(q)$;

\bigskip

\noindent $(iv)$ if $q\in [q_m,q_M]$, then $\Er^\nu$ has \emph{a minimizer $Q$ in the charge sector} $\cQ(q)$. This minimizer is not \emph{a priori} unique but its associated density $\rho_Q$ is uniquely determined. It is radially symmetric if $\nu$ is radially symmetric.
The operator $Q$ satisfies the self-consistent equation
\begin{equation}
\left\{\begin{array}{l}\displaystyle
Q+P^0_-  =  \chi_{(-\ii,\mu)}\left(D_Q\right)+\delta, \smallskip\\
D_Q  =  D^\zeta+\alpha(\rho_Q-\nu)\ast|\cdot|^{-1},
\end{array}\right.
\label{scf_equation}
\end{equation}
where $\mu\in[-1,1]$ is a Lagrange multiplier associated with the charge constraint and interpreted as a chemical potential, and $\delta$ satisfies $0\leq\delta\leq1$ and ${\rm Ran}(\delta)\subseteq \ker(D_Q-\mu)$. If $\mu\in(-1,1)$, then $\delta$ has a finite rank. If $\mu\in\{-1,1\}$, then $\delta$ is trace-class.

Moreover, $\rho_Q$ belongs to $L^1(\R^3)$ and satisfies
\begin{equation}
\fbox{$\displaystyle \int_{\R^3}\rho_Q -Z=\frac{q-Z}{1+\alpha B_\Lambda^\zeta(0)}$}
\label{int_rho}
\end{equation}
where 
$$B_\Lambda^\zeta(0)=\frac1\pi\int_0^{1}\frac{z^2-z^4/3}{(1-z^2)\left(1+\zeta\left(\frac{z^2}{\Lambda^2(1-z^2)}\right)\right)}dz=\frac{2}{3\pi}\log\Lambda+O(1)$$
if $\gH=L^2(\R^3,\C^4)$ and $\zeta\neq0$, and
$$B_\Lambda^0(0)=\frac1\pi \int_0^{\frac{\Lambda}{\sqrt{1+\Lambda^2}}}\frac{z^2-z^4/3}{1-z^2}dz=\frac{2}{3\pi}\log\Lambda-\frac{5}{9\pi}+\frac{2\log2}{3\pi}+O(1/\Lambda^2)$$
if $\gH=\gH_\Lambda$ and $\zeta=0$.
\end{thm}

\begin{figure}[h]
\centering
\input{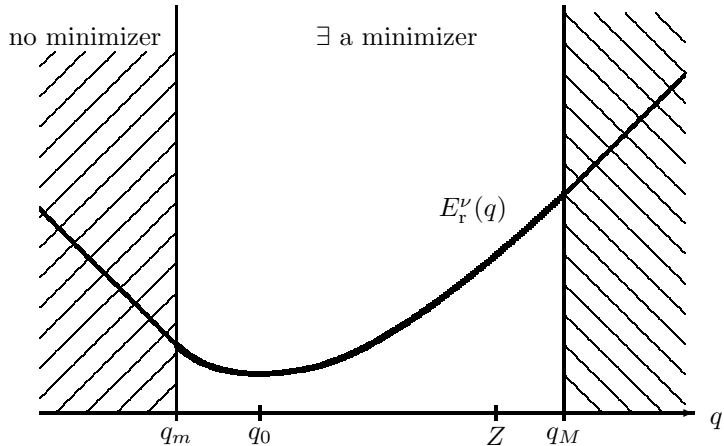}
\caption{Schematic representation of the result.}
\end{figure}

The constant $B_\Lambda^\zeta(0)$ is the value at zero of some real function $B_\Lambda^\zeta$ which will be defined later, see \eqref{def_B_Lambda_zeta} and \eqref{def_B_Lambda}. 

Equation \eqref{int_rho} has an important physical interpretation. Consider for instance a nucleus of charge $eZ$ in the vacuum, and assume that $Z$ and its distribution of charge $\nu$ are chosen to ensure that there is no pair creation from the vacuum, $\tr_{P^0_-}(\bar Q)=q_0=0$. A sufficient condition is for instance $\alpha\pi^{1/6}2^{11/6} D(\nu,\nu)^{1/2}<1$, see \cite{HLS2} and Lemma \ref{estim_apriori}. By \eqref{int_rho}, the electrostatic potential which will be observed very far away from the nucleus is $\alpha_{\rm phys}Z/|x|$ where
\begin{equation}
\alpha_{\rm phys}=\frac{\alpha}{1+\alpha B_\Lambda^\zeta(0)}.
\label{charge_renormalization}
\end{equation}
This leads to a new definition of the physical coupling constant called \emph{charge renormalization} (recall that $\alpha=e^2$). The above value of the physical charge \eqref{charge_renormalization} is very well-known in QED, see e.g. \cite[Eq. $(8)$]{Lan84} and \cite[Eq. $(7.18)$]{IZ}. This was already used and interpreted in \cite{HLS2}, in particular in connection with the large cut-off limit $\Lambda\to\ii$, in the case $\gH=\gH_\Lambda$.

The renormalized charge is only observed far away from the nucleus. Close to it, one will observe a different behavior like the oscillations of the polarization of the vacuum $\rho_{\bar Q}$. See \cite{HLSS} for an interpretation in terms of the Uehling potential.

Equation \eqref{int_rho} implies that a minimizer $Q$ in the charge sector $q\neq Z$ is never trace-class, as this would imply $\tr_{P^0_-}Q=\int_{\R^3}\rho_Q$ and contradict \eqref{int_rho}. This shows that the generalization of the reduced BDF energy $\Er^\nu$ to the Banach space $\cQ$ is mandatory, as no minimizer exists in the trace class. The mathematical difficulty that a minimizer is not trace-class is well interpreted physically in terms of charge renormalization.

When $q=Z$, it is in principle possible that a minimizer $Q$ for $E^\nu(q)$ is trace-class. We shall not investigate this question in this article.

\subsection{Ionization: an estimate on $q_m$ and $q_M$}
In the previous section, we have proved the existence of an interval $[q_m,q_M]$ for $q$ in which minimizers always exist. We now want to provide an estimate on $q_m$ and $q_M$. We do that with a specific choice for the cut-off function $\zeta$, namely $\zeta(t)=t$, which obviously satisfies our assumptions \eqref{prop_zeta_1}--\eqref{prop_zeta_3}. We give this result as an illustration: we believe that a same kind of estimate can be derived for other cut-offs. The advantage of this choice is that $D^\zeta=(-i\alp\cdot \nabla+\beta)\left(1-\frac{\Delta}{\Lambda^2}\right)$ is local. Notice in this particular case
$$B_\Lambda^\zeta(0)=\frac{2}{3\pi}\log\Lambda-\frac{5}{9\pi}+\frac{2\log2}{3\pi}+O\left(\frac{\log\Lambda}{\Lambda^2}\right).$$

\begin{thm}[Estimates on $q_m$ and $q_M$ when $Z>0$]\label{Ionization}We assume that $\gH=L^2(\R^3,\C^4)$ and $\zeta(t)=t$.   There exists universal constants $0<\theta_0<1$, $\alpha_0>0$ and $C>0$ such that the following holds. 
For any $0\leq\alpha\leq \alpha_0$, for any radial function $\nu\geq0$ in $L^1(\R^3)\cap \cC$ such that $Z=\int\nu>0$ and $\alpha D(\nu,\nu)\leq \theta_0<1$ and any cut-off $\Lambda\geq4$ such that $\alpha\log\Lambda<1/C$, the following estimate holds true:
\begin{equation}
-C\frac{Z\alpha\log\Lambda +1/\Lambda+\alpha D(\nu,\nu)}{1-C\alpha\log\Lambda}\leq q_m\leq 0=q_0,
\label{estim_q_m}
\end{equation}
\begin{equation}
Z\leq  q_M\leq \frac{2Z+C(Z\alpha\log\Lambda +1/\Lambda+\alpha D(\nu,\nu))}{1-C\alpha\log\Lambda}.
\label{estim_q_M}
\end{equation}
\end{thm}

In a nonrelativistic limit in which one takes $\alpha\to0$, $\Lambda\to\ii$ such that $\alpha\log\Lambda\to0$ and $\nu$ fixed, one obtains the usual estimate of \cite{Lieb}
$$0=q_m=q_0<Z\leq q_M\leq 2Z.$$
The proof of Theorem \ref{Ionization} is given in Section \ref{sec_proof2}. An estimate more precise than \eqref{estim_q_m} and \eqref{estim_q_M} is contained in our proof but we do not state it here.

\section{Proofs}\label{sec_proof}
\subsection{Proof of Proposition \ref{prop_def_rho}}\label{proof_prop_def_rho}
When $\gH=\gH_\Lambda$ and $\zeta=0$, Proposition \ref{prop_def_rho} is contained in \cite[Lemma 1]{HLS3}. Hence we only treat the case $\gH=L^2(\R^3,\C^4)$.
Consider some $Q\in \cQ$ and $V\in\cC'\cap L^\ii(\R^3)$. We have $(QV)^{++}=Q^{++}VP^0_++Q^{+-}[V,P^0_+]P^0_+$ and $(QV)^{--}=Q^{--}VP^0_-+Q^{-+}[V,P^0_-]P^0_-$. We first give an estimate on the commutator $[V,P^0_-]$.

\begin{lemma}\label{lemma_estim_commut}
We assume that $\gH=\gH_\Lambda$ and $\zeta=0$ or $\gH=L^2(\R^3,\C^4)$ and $\zeta$ satisfies \eqref{prop_zeta_1}--\eqref{prop_zeta_3}. 
We have for all $\tau\geq 1/2$ and all $p\geq 2$
$$\forall V,\qquad\norm{|D^\zeta|^{-\tau}[V,P^0_-]}_{\gS_p(\gH)}\leq C\norm{\nabla V}_{L^p(\R^3)}$$
where the constant $C$ is independent of $\zeta$ (hence of $\Lambda$) if $\tau>1/2$ or $p>2$.
\end{lemma}
\begin{proof}[Proof of Lemma \ref{lemma_estim_commut}]
Using Cauchy's formula, we infer
\begin{eqnarray*}
[V,P^0_-] & = & \frac{1}{2\pi}\int_{-\ii}^\ii  \left(\frac{1}{D^0+i\eta}V-V\frac{1}{D^0+i\eta}\right)d\eta\\
 & = & \frac{1}{2\pi}\int_{-\ii}^\ii \frac{1}{D^0+i\eta}[V,D^0+i\eta]\frac{1}{D^0+i\eta}d\eta\\
 & = & -\frac{i}{2\pi}\int_{-\ii}^\ii \frac{1}{D^0+i\eta}\alp\cdot(\nabla V)\frac{1}{D^0+i\eta}d\eta
\end{eqnarray*}
(recall $P^0_-$ does not depend on $\zeta$). Hence, by means of
$$\norm{\frac{1}{D^0+i\eta}}\leq\frac{1}{\sqrt{1+\eta^2}}$$
we obtain
$$\norm{|D^\zeta|^{-\tau}[V,P^0_-]}_{\gS_p(\gH)}\leq\frac{1}{2\pi}\int_{-\ii}^\ii \norm{\frac{1}{|D^\zeta|^{\tau}(D^0+i\eta)}\alp\cdot(\nabla V)}_{\gS_p(\gH)}\frac{d\eta}{\sqrt{1+\eta^2}}.$$
Next we use the Kato-Seiler-Simon inequality \eqref{KSS} and obtain by \eqref{prop_zeta_2}
\begin{multline}
\norm{|D^\zeta|^{-\tau}[V,P^0_-]}_{\gS_p(\gH)}\leq C\norm{\nabla V}_{L^p(\R^3)}\times\\
\times\int_{-\ii}^\ii \norm{\frac{1}{(1+|\cdot|)^{\tau}(1+\epsilon|\cdot|^{\epsilon}/\Lambda^{\epsilon})^\tau\sqrt{1+|\cdot|^2+\eta^2}}}_{L^p(\R^3)}\frac{d\eta}{\sqrt{1+\eta^2}}
\end{multline}
which allows to conclude.
\end{proof}

We consider first for $(QV)^{++}$ and use Lemma \ref{lemma_estim_commut} with $p=2$ and $\tau=1/2$,
\begin{align*}
\norm{Q^{+-}[V,P^0_+]P^0_+}_{\gS_1(\gH)}&\leq C\norm{Q^{+-}|D^\zeta|^{1/2}}_{\gS_2(\gH)}\norm{\nabla V}_{L^2(\R^3)}\\
&\leq C\norm{Q|D^\zeta|^{1/2}}_{\gS_2(\gH)}\norm{V}_{\cC'}. 
\end{align*}
Similarly we have
$$\norm{Q^{++}VP^0_+}_{\gS_1(\gH)}\leq \norm{Q^{++}|D^\zeta|^{1/2}}_{\gS_1(\gH)}\norm{|D^\zeta|^{-1/2}V}_{\gS_\ii(\gH)}.$$
On the other hand, we have by the Kato-Seiler-Simon inequality \eqref{KSS}
$$\norm{|D^\zeta|^{-1/2}V}_{\gS_\ii(\gH)}\leq \norm{|D^\zeta|^{-1/2}V}_{\gS_6(\gH)}\leq C \norm{V}_{L^6}\norm{|D^\zeta(\cdot)|^{-1}}_{L^3(\R^3)}^2\leq C\norm{V}_{L^6}.$$
where we have used  Assumption \eqref{prop_zeta_2} on $\zeta$ and $\norm{|D^\zeta(\cdot)|^{-1}}_{L^3(\R^3)}<\ii$. Hence 
$$\norm{|D^\zeta|^{-1/2}V}_{\gS_\ii(\gH)}\leq C\norm{\nabla V}_{L^2\R^3)}=C\norm{V}_{\cC'}$$
by the critical Sobolev embedding $H^1(\R^3)\hookrightarrow L^6(\R^3)$. As a conclusion, 
$$\left|\tr(QV)^{++}\right|\leq\norm{(QV)^{++}}_{\gS_1(\gH)} \leq C\norm{V}_{\cC'}\norm{Q}_{\cQ}.$$
The proof is the same for $(QV)^{--}$.\qed

\subsection{Proof of Theorem \ref{exists}}\label{proof_thm_exists}
\subsubsection*{\bf Step 1: Existence of a minimizer if some HVZ conditions hold.}
Let us start with the analogue of \cite[Lemma 3]{HLS3}.
\begin{lemma}We assume that $\gH=\gH_\Lambda$ and $\zeta=0$, or that $\gH=L^2(\R^3,\C^4)$ and $\zeta$ satisfies \eqref{prop_zeta_1}--\eqref{prop_zeta_3}.  Let be $\alpha\geq0$, $\Lambda>0$ and $\nu\in L^1(\R^3)\cap \cC$. We have the following estimate
\begin{equation}
|q|-\frac\alpha2 D(\nu,\nu) \leq E_{\rm r}^\nu(q) \leq |q|.
\label{estim_q}
\end{equation}
In particular we get for $\nu=0$ and for any $q\in\R$,
$$ E_{\rm r}^0(q) = |q|.$$
\end{lemma}
\begin{proof}
It suffices to follow the proof of \cite[Lemma 3]{HLS3}. 
\end{proof}

Next we state a result analogous to \cite[Theorem 3]{HLS3}.
\begin{thm}[A dissociation criterion]\label{HVZ}We assume that $\gH=\gH_\Lambda$ and $\zeta=0$, or that $\gH=L^2(\R^3,\C^4)$ and $\zeta$ satisfies \eqref{prop_zeta_1}--\eqref{prop_zeta_3}. Let be $\alpha\geq0$, $\Lambda>0$ and $\nu\in L^1(\R^3)\cap \cC$.
The following two conditions are equivalent

\medskip

\noindent$\rm (H_1)$\quad  $E_{\rm r}^\nu(q) < E^\nu_{\rm r}(q')+|q-q'|$ for any $q'\neq q$;

\medskip

\noindent$\rm (H_2)$\quad each minimizing sequence $(Q_n)_{n\geq1}$ for $E_{\rm r}^\nu(q)$ is precompact in $\cQ$ and converges, up to a subsequence, to a minimizer $Q$ of $E^\nu_{\rm r}(q)$. 

\medskip

When it exists, such a minimizer $Q$ satisfies the self-consistent equation
\begin{equation}
\left\{\begin{array}{l}\displaystyle
Q+P^0_-  =  \chi_{(-\ii,\mu)}\left(D_Q\right)+\delta, \smallskip\\
D_Q  =  D^\zeta+\alpha(\rho_Q-\nu)\ast|\cdot|^{-1},
\end{array}\right.
\label{scf_equation2}
\end{equation}
where $\mu\in[-1,1]$ is a Lagrange multiplier associated with the charge constraint and interpreted as a chemical potential, and $\delta$ is a self-adjoint operator satisfying $0\leq\delta\leq1$ and ${\rm Ran}(\delta)\subseteq \ker(D_Q-\mu)$. The operator $\delta$ is finite rank if $\mu\in(-1,1)$ and trace-class if $\mu\in\{-1,1\}$.
\end{thm}
\begin{remark}\rm
Like in \cite[Prop. 8]{HLS3}, it can be proved that 
\begin{equation}
\forall q,q'\in\R,\qquad E_{\rm r}^\nu(q)\leq E_{\rm r}^\nu(q')+|q-q'|.
\label{large_HVZ}
\end{equation}
In particular this implies that $q\mapsto E^\nu_{\rm r}(q)$ is Lipschitz. 
\end{remark}
\begin{proof}
The proof of Theorem \ref{HVZ} is an adaptation of previous works and it will not be detailed here. In the case of the sharp cut-off $\gH=\gH_\Lambda$ and $\zeta=0$, this is contained in the proof of \cite[Theorem 3]{HLS3}. In the smooth cut-off case $\gH=L^2(\R^3,\C^4)$ with $\zeta\neq0$, it suffices to follow the proof given in the crystal case in \cite{CDL}. Notice many commutator estimates proved in \cite{CDL} (like \cite[Lemma 11]{CDL}) are derived using the regularity of $\zeta$ and the fact that its derivatives grow at most algebraically as expressed by our assumptions \eqref{prop_zeta_1}--\eqref{prop_zeta_3}.

The proof that a minimizer $Q$ satisfies Equation \eqref{scf_equation2} is the same as in \cite[Theorem 3]{HLS2} and \cite[Proposition 2]{HLS3}. Finally, $\delta$ is finite-rank if $\mu<1$ because the essential spectrum of $D_Q$ is the same as that of $D^0$ by Lemma \ref{lem_weyl}. If $\mu=1$, let us recall \cite{HLS1,HLS2} that $Q_{\rm vac}:=\chi_{(-\ii,0)}(D_Q)-P^0_-\in\gS_2(\gH)$ (see Lemma \ref{prop_Q_k}). By \cite[Lemma 2]{HLS1}, we have $Q_{\rm vac}\in\gS_1^{P^0_-}(\gH)$. Hence we deduce $Q-Q_{\rm vac}\in\gS_1^{P^0_-}(\gH)$ which tells us that $Q-Q_{\rm vac}$ and $\delta$ are trace-class because they are nonnegative.
\end{proof}

\begin{prop}\label{prop_unique_density_SU2}
Minimizers of $E^\nu_{\rm r}(q)$ are not necessarily unique, but the density $\rho_Q$ is itself uniquely defined. 
If $\nu$ is radially symmetric, then so does $\rho_Q$.
\end{prop}
\begin{proof}
Note $Q\in\cQ\to\Er^\nu(Q)$ is convex but not strictly convex. The term $f\to D(f,f)$ is strictly convex but the map $Q\to \rho_Q$ is not one-to-one. This, however, implies that the density $\rho_Q$ of a minimizer is uniquely determined, meaning that if $Q_1$ and $Q_2$ are two minimizers of $E^\nu_{\rm r}(q)$, then necessarily $\rho_{Q_1}=\rho_{Q_2}$. 

Next we recall that any unitary matrix $U\in SU_2$ can be written $U=e^{-i\theta n\cdot\sigma}$ where $\theta\in[0,2\pi)$ and $n$ is a unit vector in $\R^3$. There is an onto morphism which to any such $U$ associates the rotation $R_{\theta,n}$ in $\R^3$ of angle $\theta$ around the axis $n$. The group $SU_2$ acts on $4$-spinors in $L^2(\R^3,\C^4)$ as follows:
$$(U\cdot \psi)(x):=\left(\begin{matrix}U & 0\\ 0 & U\end{matrix}\right)\psi(R_{\theta,n}^{-1}x).$$
It is well-known \cite{Thaller} that the Dirac operator $D^0$ is invariant under this action. As $D^\zeta$ is equal to $D^0$ multiplied by a radial function in the Fourier domain, $D^\zeta$ is also invariant. When $\nu$ is a radial function, we hence have
$\Er^\nu(Q)=\Er^\nu(UQU^{-1})$ and $\tr_{P^0_-}(Q)=\tr_{P^0_-}(UQU^{-1})$ for any $Q\in\cK$ and any $U\in SU_2$. This means that if $Q$ is a minimizer for $E_{\rm r}^\nu(q)$, then $UQU^{-1}$ is also a minimizer. As $\rho_{UQU^{-1}}(x)=\rho_Q(R_{\theta,n}^{-1}x)$, we deduce by uniqueness that $\rho_Q$ is a radial function.
\end{proof}

\subsubsection*{\bf Step 2: The density of a solution is in $L^1$.}
We prove the important
\begin{thm}[The density of a solution is in $L^1$]\label{thm_L1}
We assume that $\gH=\gH_\Lambda$ and $\zeta=0$, or that $\gH=L^2(\R^3,\C^4)$ and $\zeta$ satisfies \eqref{prop_zeta_1}--\eqref{prop_zeta_3}. Let be $\alpha\geq0$, $\Lambda>0$, $\nu\in L^1(\R^3)\cap \cC$ and denote $Z=\int_{\R^3}\nu\in\R$. If $Q\in\cQ(q)$ satisfies the self-consistent equation \eqref{scf_equation}, then $\rho_Q\in L^1(\R^3)$ and 
\begin{equation}
\int_{\R^3}\rho_Q -Z=\frac{q-Z}{1+\alpha B^\zeta_\Lambda(0)}.
\label{int_rho2}
\end{equation}
\end{thm}

\begin{proof}
We shall do more than proving that $\rho_Q\in L^1(\R^3)$. Namely, we shall provide a precise estimate on $\norm{\rho_{Q_{\rm vac}}}_{L^1(\R^3)}$ needed for the proof of Theorem \ref{Ionization}.

Let $Q\in\cQ(q)$ satisfying the self-consistent equation $Q=\chi_{(-\ii,\mu)}(D_Q)-P^0_-+\delta$ where $\delta$ is a trace-class self-adjoint operator with ${\rm Ran}(\delta)\subseteq\ker(D_Q-\mu)$, and $D_Q$ is the mean-field operator: $$D_Q=D^\zeta+\alpha(\rho_Q-\nu)\ast|\cdot|^{-1}.$$
 Recall that by Lemma \ref{lem_weyl}, $\sigma_{\rm ess}(D_Q)=\sigma_{\rm ess}(D^\zeta)=(-\ii,-1]\cup[1,\ii)$, i.e. that $\sigma(D_Q)\cap(-1,1)$ contains eigenvalues of finite multiplicity, possibly accumulating at $-1$ or 1. For the sake of simplicity, we shall assume that $0\notin\sigma(D_Q)$. The following proof can be adapted if $0\in\sigma(D_Q)$ by integrating on a line $\epsilon+i\eta$ instead of $i\eta$ in the integrals below. We introduce the notation
$$Q_{\rm vac}:=\chi_{(-\ii,0]}(D_Q)-P^0_-,\qquad \gamma=Q-Q_{\rm vac}.$$
Notice $\gamma\in\gS_1(\gH)$. We recall that $Q_{\rm vac}^{++}=P^0_+Q_{\rm vac}P^0_+$, $Q_{\rm vac}^{--}=P^0_-Q_{\rm vac}P^0_-\in \gS_1(\gH_\Lambda)$. Hence we have to prove that $\rho_{Q^{+-}_{\rm vac}+Q^{-+}_{\rm vac}}$ belongs to $L^1(\R^3)$, which we will do by a bootstrap argument on the self-consistent equation.

We can use Cauchy's formula as in \cite{HLS1}
\begin{equation}
Q_{\rm vac}  =  -\frac1{2\pi}\int_{-\ii}^\ii\left(\frac{1}{D_Q+i\eta}-\frac{1}{D^\zeta+i\eta}\right)d\eta= \sum_{k=1}^3\alpha^k Q_k+\alpha^4 Q'_4
\end{equation}
with
\begin{equation}
 Q_k=(-1)^{k+1}\frac1{2\pi}\int_{-\ii}^\ii\frac{1}{D^\zeta+i\eta}\left(\phi'_Q\frac{1}{D^\zeta+i\eta}\right)^kd\eta,
\label{Q_k}
\end{equation}
$$Q'_4=-\frac1{2\pi}\int_{-\ii}^\ii\left(\frac{1}{D^\zeta+i\eta}\phi'_Q\right)^2\frac{1}{D_Q+i\eta}\left(\phi'_Q\frac{1}{D^\zeta+i\eta}\right)^2d\eta$$
and where we have used the notation $\phi'_Q=(\rho_Q-\nu)\ast|\cdot|^{-1}$. By Furry's Theorem, it is known that $\rho_{Q_2}=0$, see \cite[page 547]{HLS1}. 

\begin{lemma}\label{prop_Q_k}
Let be $0\leq\tau<1/2$. There exists a universal constant $C$  such that the following hold:
\begin{equation*}
\norm{|D^\zeta|^\tau Q_1}_{\gS_2(\gH)}\leq C\norm{\rho_Q-\nu}_\cC,
\label{lem_estim_Q_1} 
\end{equation*}
\begin{equation*}
\norm{|D^\zeta|^{1/2+\tau} Q_2}_{\gS_{3/2}(\gH)}\leq C\norm{\rho_Q-\nu}^2_\cC,\quad
\norm{|D^\zeta|Q_3}_{\gS_{6/5}(\gH)}\leq C\norm{\rho_Q-\nu}^3_\cC,
\label{lem_estim_Q_2et3} 
\end{equation*}
\begin{equation*}
\norm{|D^\zeta|^\tau Q'_4|D^\zeta|^\tau}_{\gS_1(\gH)}\leq C\left(\norm{\rho_Q-\nu}^4_\cC+\alpha\norm{\rho_Q-\nu}^5_\cC+\alpha^2\frac{\norm{\rho_Q-\nu}^6_\cC}{{\rm dist}(\sigma(D_Q),0)}\right).
\label{lem_estim_Q_4} 
\end{equation*}
\end{lemma}
\begin{proof}
By the residuum formula, we have $Q^{++}_1=Q^{--}_1=0$. On the other hand,
$$Q_1^{+-}=\frac1{2\pi}\int_{-\ii}^\ii\frac{P^0_+}{D^\zeta+i\eta}\phi'_Q\frac{P^0_-}{D^\zeta+i\eta}d\eta=\frac1{2\pi}\int_{-\ii}^\ii\frac{P^0_+}{D^\zeta+i\eta}[\phi'_Q,P^0_-]\frac{P^0_-}{D^\zeta+i\eta}d\eta.$$
Hence using
$$\norm{\frac{1}{D^\zeta+i\eta}}\leq\frac{1}{E(\eta)},\qquad \norm{\frac{|D^\zeta|^{1/2+\tau}}{D^\zeta+i\eta}}\leq\frac{1}{E(\eta)^{1/2-\tau}}$$
where $E(\eta):=\sqrt{1+\eta^2}$, and using also Lemma \ref{lemma_estim_commut}, we obtain
$$\norm{|D^\zeta|^\tau Q^{+-}}_{\gS_2(\gH)}\leq C\norm{\nabla\phi'_Q}_{L^2(\R^3)}\int_{-\ii}^\ii\frac{d\eta}{E(\eta)^{3/2-\tau}}=C\norm{\rho_Q-\nu}_\cC$$
since $\norm{\nabla\phi'_Q}_{L^2(\R^3)}=\norm{\rho_Q-\nu}_\cC$.

We then turn to $Q_2$, inserting $1=P^0_-+P^0_+$ in \eqref{Q_k}. We first notice that by the residuum formula,
$$\int_{-\ii}^\ii\frac{P^0_+}{D^\zeta+i\eta}\left(\phi'_Q\frac{P^0_+}{D^\zeta+i\eta}\right)^2d\eta=\int_{-\ii}^\ii\frac{P^0_-}{D^\zeta+i\eta}\left(\phi'_Q\frac{P^0_-}{D^\zeta+i\eta}\right)^2d\eta=0.$$
For the other terms, we write for instance
\begin{multline}
\int_{-\ii}^\ii\frac{P^0_+}{D^\zeta+i\eta}\phi'_Q\frac{P^0_-}{D^\zeta+i\eta}\phi'_Q\frac{P^0_-}{D^\zeta+i\eta}d\eta\\
=\int_{-\ii}^\ii\frac{P^0_+}{D^\zeta+i\eta}[\phi'_Q,P^0_-]\frac{P^0_-}{D^\zeta+i\eta}\phi'_Q\frac{P^0_-}{D^\zeta+i\eta}d\eta
\label{Q_2+--}
\end{multline}
as we did before. We recall that $\phi'_Q\in L^6(\R^3)$ by the Sobolev inequality. Hence, by \eqref{KSS}
\begin{equation}
\norm{\phi'_Q\frac{1}{D^\zeta+i\eta}}_{\gS_6(\gH)} \leq  \frac{C}{E(\eta)^{1/2}}\norm{\nabla\phi'_Q}_{L^2(\R^3)}.
\label{estim_S6}
\end{equation}
Using again Lemma \ref{lemma_estim_commut}, we obtain
\begin{align*}
&\norm{|D^\zeta|^{1/2+\tau}\int_{-\ii}^\ii\frac{P^0_+}{D^\zeta+i\eta}\phi'_Q\frac{P^0_-}{D^\zeta+i\eta}\phi'_Q\frac{P^0_-}{D^\zeta+i\eta}d\eta}_{\gS_{3/2}(\gH)}\\
&\qquad\qquad\leq \norm{[\phi'_Q,P^0_-]|D^\zeta|^{\tau/2-3/4}}_{\gS_{2}(\gH)} \int_{-\ii}^\ii\norm{\phi'_Q\frac{1}{D^\zeta+i\eta}}_{\gS_6(\gH)}\frac{d\eta}{E(\eta)^{3/4-\tau/2}}\\
&\qquad\qquad \leq C\norm{\rho_Q-\nu}_{\cC}^2\int_{-\ii}^\ii\frac{d\eta}{E(\eta)^{5/4-\tau/2}}.
\end{align*}
The proof is the same for all the other terms. 

The same method can be applied to $Q_3$. Let us treat for instance
\begin{multline}
A:=|D^\zeta|\int_{-\ii}^\ii\frac{P^0_+}{D^\zeta+i\eta}\phi'_Q\frac{P^0_-}{D^\zeta+i\eta}\left(\phi'_Q\frac{P^0_-}{D^\zeta+i\eta}\right)^2d\eta\\
=\int_{-\ii}^\ii\frac{P^0_+|D^\zeta|}{D^\zeta+i\eta}[\phi'_Q,P^0_-]\frac{P^0_-}{D^\zeta+i\eta}\left(\phi'_Q\frac{P^0_-}{D^\zeta+i\eta}\right)^2d\eta. 
\label{Q_3+--}
\end{multline}
Applying the above method with
$$\norm{[\phi'_Q,P^0_-]\;|D^\zeta|^{-3/4}}_{\gS_2(\gH)}\leq C\norm{\nabla\phi'_Q}_{L^2(\R^3)}$$
by Lemma \ref{lemma_estim_commut}, we obtain
\begin{equation}
\norm{A}_{\gS_{6/5}(\gH)}\leq C\norm{\rho_Q-\nu}^3_\cC \int_{-\ii}^\ii\frac{d\eta}{(1+\eta^2)^{5/8}}. 
\label{Q_3+--2}
\end{equation}
The argument is of course the same for all the other terms.

Finally, we expand further $Q'_4$ to the 6th order:
$Q'_4=Q_4+Q_5+Q'_6$ where $Q_4$ and $Q_5$ are given by \eqref{Q_k} and
$$Q'_6=-\frac1{2\pi}\int_{-\ii}^\ii\left(\frac{1}{D^\zeta+i\eta}\phi'_Q\right)^3\frac{1}{D_Q+i\eta}\left(\phi'_Q\frac{1}{D^\zeta+i\eta}\right)^3d\eta.$$
On the one hand, we know that $|D_Q+i\eta|\geq {\rm dist}(\sigma(D_Q),0)$), and therefore,
$$\norm{(D_Q+i\eta)^{-1}}\leq {\rm dist}(\sigma(D_Q),0)^{-1}.$$
On the other hand, we can use \eqref{estim_S6}  and
$$\norm{\frac{|D^\zeta|^{\tau}}{D^\zeta+i\eta}\phi'_Q}_{\gS_6(\gH)} \leq  C \norm{|D^0(\cdot)|^{\tau-1}}_{L^6(\R^3)}\norm{\phi'_Q}_{L^6(\R^3)}$$
to estimate $|D^\zeta|^\tau Q_6'|D^\zeta|^\tau$. The terms $Q_4$ and $Q_5$ are treated like $Q_2$ and $Q_3$.
\end{proof}

\begin{lemma}\label{polynomial_argument} Let be $0\leq\tau<1/2$. There exists a universal constant $C$ such that
\begin{equation}
\norm{|D^\zeta|^\tau Q_2^{\pm\pm}|D^\zeta|^\tau }_{\gS_1(\gH)}\leq C\left(\norm{\rho_Q-\nu}_\cC^2+\alpha^2\norm{\rho_Q-\nu}_\cC^4\right),
\label{estim_Q2++} 
\end{equation}
\begin{equation}
\norm{|D^\zeta|^\tau Q_3^{\pm\pm}|D^\zeta|^\tau }_{\gS_1(\gH)}\leq C\left(\norm{\rho_Q-\nu}_\cC^3+\alpha^2\norm{\rho_Q-\nu}_\cC^5\right).
\label{estim_Q3++} 
\end{equation}
\end{lemma}
\begin{proof}
Consider the operator 
$$D(t):=D^\zeta+t\frac{(\rho_Q-\nu)\ast|\cdot|^{-1}}{\norm{\rho_Q-\nu}_\cC}.$$
Since $\rho_Q-\nu\in\cC$, we can use \eqref{estim_weyl} to deduce that there exists a universal constant $t_0>0$ such that $|D(t)|\geq1/2$ for all $t\in[-t_0,t_0]$.
Next we introduce
$$Q(t):=\chi_{(-\ii;0]}(D(t))-P^0_-.$$
We can write as before
$$Q(t)=\sum_{k=1}^3\frac{t^k}{\norm{\rho_Q-\nu}_\cC^k} Q_k +\frac{t^4}{\norm{\rho_Q-\nu}_\cC^4}Q'_4(t) $$
where $Q_k$ are defined as above and with this time
$$Q'_4(t)=-\frac1{2\pi}\int_{-\ii}^\ii\left(\frac{1}{D^\zeta+i\eta}\phi'_Q\right)^2\frac{1}{D(t)+i\eta}\left(\phi'_Q\frac{1}{D^\zeta+i\eta}\right)^2d\eta.$$
Following the method of Lemma \ref{prop_Q_k} and using $|D(t)|\geq 1/2$, we can prove that 
$$\norm{|D^\zeta|^\tau Q'_4(t)|D^\zeta|^\tau }_{\gS_2(\gH)}\leq C\norm{\rho_Q-\nu}_\cC^4,$$
$$\norm{|D^\zeta|^\tau Q'_4(t)|D^\zeta|^\tau }_{\gS_1(\gH)}\leq C\left(\norm{\rho_Q-\nu}_\cC^4+\alpha^2\norm{\rho_Q-\nu}_\cC^6\right).$$
Next, the estimates of Lemma \ref{prop_Q_k} imply that 
$$\norm{|D^\zeta|^\tau Q(t)}_{\gS_2(\gH)}\leq C$$
for all $t\in[-t_0,t_0]$. But as $Q(t)$ is a difference of two projectors, we have $Q(t)^2=Q(t)^{++}-Q(t)^{--}\in\gS_1(\gH)$. Thus
$$\norm{|D^\zeta|^\tau Q(t)^{++}|D^\zeta|^\tau }_{\gS_1(\gH)}+\norm{|D^\zeta|^\tau Q(t)^{--}|D^\zeta|^\tau }_{\gS_1(\gH)}\leq C.$$
Finally
$$\frac{t^2}{\norm{\rho_Q-\nu}_\cC^2}Q_2^{++}+\frac{t^3}{\norm{\rho_Q-\nu}_\cC^3}Q_3^{++}=Q(t)^{++}-\frac{t^4}{\norm{\rho_Q-\nu}_\cC^4}Q'_4(t)^{++}$$
which gives the result when applied to $t=t_0$ and $-t_0$. 
\end{proof}

\begin{lemma}\label{density_Q3} Let be $2\leq p\leq 6$. There exists a universal constant $C$ such that 
$$\norm{\rho_{(Q_3)^{+-}}\ast|\cdot|^{-1}}_{L^p(\R^3)}+\norm{\rho_{(Q_3)^{-+}}\ast|\cdot|^{-1}}_{L^p(\R^3)}\leq C\norm{\rho_Q-\nu}_{\cC}^3.$$
\end{lemma}
\begin{proof}
By Lemma \ref{prop_Q_k}, $Q_3|D^\zeta|\in \gS_{6/5}(\gH)$, hence $Q_3|D^\zeta|\in \gS_{q}(\gH)$ for all $q\geq 6/5$ and
\begin{equation}
 \norm{Q_3^{+-}|D^\zeta|}_{\gS_{q}(\gH)}\leq\norm{Q_3|D^\zeta|}_{\gS_{q}(\gH)}\leq C\norm{\rho_Q-\nu}^3_\cC.
\label{philippe1}
\end{equation}
Let us choose a test function $V$ in the Schwartz class. We have
\begin{multline}
|\tr((Q_3)^{+-}V)|  =  |\tr((Q_3)^{+-}P^0_-V P^0_+)| = |\tr\left((Q_3)^{+-}|D^\zeta|\;|D^\zeta|^{-1}[P^0_-,V]\right)|\\
 \leq  \norm{(Q_3)^{+-}|D^\zeta|}_{\gS_{q}(\gH)}\norm{|D^\zeta|^{-1}[P^0_-,V]}_{\gS_{q'}(\gH)}\label{philippe2}
\end{multline}
for all $q\geq6/5$ and $q'=q/(q-1)$. Then we use Lemma \ref{lemma_estim_commut} which tells us that
$$\norm{|D^\zeta|^{-1}[P^0_-,V]}_{\gS_{q'}(\gH)}\leq C\norm{\nabla V}_{L^{q'}(\R^3)}$$
provided $q'\geq 2$.
Finally by the Sobolev inequality and Riesz operator theory
$$\norm{\nabla V}_{L^{q'}(\R^3)}\leq C\norm{D^2V}_{L^{p^*}(\R^3)} \leq C' \norm{\Delta V}_{L^{p^*}(\R^3)}$$
for $p^*=3q'/(3+q')$, $2\leq q'\leq 6$. 
Summarizing,  by \eqref{philippe1} and \eqref{philippe2},
$$|\tr((Q_3)^{+-}V)|\leq C\norm{\rho_Q-\nu}_{\cC}^3\norm{\Delta V}_{L^{p^*}(\R^3)}$$
for any $6/5\leq p^*\leq 2$. By duality, this proves that for any $2\leq p\leq 6$
$$\norm{\rho_{(Q_3)^{+-}}\ast|\cdot|^{-1}}_{L^p(\R^3)}\leq C\norm{\rho_Q-\nu}_{\cC}^3.$$
\end{proof}

\begin{lemma}\label{density_Q3bis} Let be $3< p<\ii$. There exists a universal constant $C$ such that 
$$\norm{\rho_{(Q_3)^{\pm\pm}}\ast|\cdot|^{-1}}_{L^p(\R^3)}\leq C\left(\norm{\rho_Q-\nu}_{\cC}^3+\alpha^2\norm{\rho_Q-\nu}_{\cC}^5\right),$$
$$\norm{\rho_{Q_4'}\ast|\cdot|^{-1}}_{L^p(\R^3)}\leq C\left(\norm{\rho_Q-\nu}^4_\cC+\alpha\norm{\rho_Q-\nu}^5_\cC+\alpha^2\frac{\norm{\rho_Q-\nu}^6_\cC}{{\rm dist}(\sigma(D_Q),0)}\right).$$
\end{lemma}
\begin{proof}
We argue as above, taking some $V$ in the Schwartz class. We have
$$|\tr(Q_3^{++}V)|\leq \norm{|D^\zeta|^\tau Q_3^{++}|D^\zeta|^\tau}_{\gS_1(\gH)}\norm{|D^\zeta|^{-\tau} V|D^\zeta|^{-\tau}}_{\gS_q(\gH)}$$
for any $q\geq1$ and $\tau<1/2$. Then by the Kato-Seiler-Simon inequality \eqref{KSS}
\begin{equation*}
 \norm{|D^\zeta|^{-\tau} V|D^\zeta|^{-\tau}}_{\gS_q(\gH)}\leq C \norm{E(\cdot)^{-2\tau}}_{L^q(\R^3)}\norm{V}_{L^q(\R^3)}
\end{equation*}
which makes sense as soon as $q>3$ and $1/2-\tau$ is small enough.
The rest follows from the Sobolev embedding like in the proof of Lemma \ref{density_Q3}.
\end{proof}

We now estimate $\rho_{Q_{\rm vac}}$ using the self-consistent equation. First we recall that $\rho_{Q_2}=0$ and that $Q_1=(Q_1)^{+-}+(Q_1)^{-+}$ can be explicitly computed \cite{HLS1} yielding
\begin{multline}
\widehat{\rho_{Q_{\rm vac}}}(k)=-\alpha B^\zeta_\Lambda(k)\left(\widehat{\rho_{Q_{\rm vac}}}(k)+\widehat{\rho_{\gamma}}(k)-\widehat{\nu}(k)\right)\\
+\alpha^4\widehat{\rho_{Q'_4}}(k)+\alpha^3\widehat{\rho_{Q^{++}_3}}(k)+\alpha^3\widehat{\rho_{Q_3^{--}}}(k)
+\alpha^3\widehat{\rho_{Q_3^{+-}}}(k)+\alpha^3\widehat{\rho_{Q_3^{-+}}}(k)\label{SCF_eq_rho}
\end{multline}
with
\begin{multline}
B^\zeta_\Lambda(k)=-\frac{1}{\pi^2 |k|^2}\int_{\R^3} 
\frac{(\ell+k/2)\cdot(\ell-k/2)+1-E(\ell+k/2)E(\ell-k/2)}{E(\ell+k/2)E(\ell-k/2)}\times\\
\times \frac{1}{E(\ell+k/2)\left(1+\zeta\left(\frac{|\ell+k/2|^2}{\Lambda^2}\right)\right)+E(\ell-k/2)\left(1+\zeta\left(\frac{|\ell-k/2|^2}{\Lambda^2}\right)\right)}d\ell
\label{def_B_Lambda_zeta}
\end{multline}
when $\gH=L^2(\R^3,\C^4)$ and $\zeta$ satisfies \eqref{prop_zeta_1}--\eqref{prop_zeta_3}, and
\begin{equation}
B^0_\Lambda(k)=-\frac{1}{\pi^2 |k|^2}\int_{\substack{|\ell+k/2|\leq\Lambda,\\ |\ell-k/2|\leq\Lambda}} 
\frac{(\ell+k/2)\cdot(\ell-k/2)+1-E(\ell+k/2)E(\ell-k/2)}{E(\ell+k/2)E(\ell-k/2)(E(\ell+k/2)+E(\ell-k/2))}d\ell.
\label{def_B_Lambda}
\end{equation} 
when $\gH=\gH_\Lambda$ and $\zeta=0$.
Notice that in both cases $B^\zeta_\Lambda$ is a radial function. Also $B^0_\Lambda$ has its support in $B(0,2\Lambda)$.
\begin{remark}\rm
There is a small mistake in the domain of integration of the definition of $B^0_\Lambda$ in \cite[Eq. (40)]{HLS1}. This does not change the analysis of \cite{HLS1} but is important for the present study.
\end{remark}

Many properties of $B^\zeta_\Lambda$ and $B^0_\Lambda$ are given in Appendix. 
Let us define $b_\Lambda^\zeta$ by
\begin{equation}\label{def_b_Lambda}
\widehat{b_\Lambda^\zeta}(k)=(2\pi)^{3/2}\frac{\alpha B^\zeta_\Lambda(k)}{1+\alpha B^\zeta_\Lambda(k)}.
\end{equation}
In both cases $\gH=\gH_\Lambda$ with $\zeta=0$, and $\gH=L^2(\R^3,\C^4)$ with $\zeta$ satisfying \eqref{prop_zeta_1}--\eqref{prop_zeta_3}, we prove in Appendix that $b^\zeta_\Lambda$ is a smooth function belonging to $L^1(\R^3)$, see Propositions \ref{prop_b_Lambda_0} and \ref{prop_b_Lambda_zeta}. In the rest of the proof, we use the notation
\begin{equation}
 I_\Lambda:=\int_{\R^3}|b^\zeta_\Lambda(x)|\,dx<\ii.
\label{def_int_b_Lambda}
\end{equation}
Equation \eqref{SCF_eq_rho} can be rewritten as
\begin{equation}\label{decomp_rho}
\rho_{Q_{\rm vac}}=b_\Lambda^\zeta\ast(\nu-\rho_\gamma-\rho_1)+\rho_1+\rho_2-b_\Lambda^\zeta\ast\rho_2,
\end{equation}
$$\rho_1=\alpha^4\rho_{Q'_4}+\alpha^3\rho_{Q^{++}_3}+\alpha^3\rho_{Q_3^{--}}\in L^1(\R^3),\qquad \rho_2=\alpha^3\rho_{Q^{+-}_3}+\alpha^3\rho_{Q_3^{-+}}.$$
By Lemma \ref{density_Q3}, Lemma \ref{density_Q3bis} and \eqref{def_int_b_Lambda}
$$\norm{(\rho_2-b_\Lambda^\zeta\ast\rho_2)\ast|\cdot|^{-1}}_{L^4(\R^3)}\leq C(1+I_\Lambda)\alpha^3\norm{\rho_Q-\nu}_{\cC}^3,$$
\begin{multline*}
\norm{(\rho_1-b_\Lambda^\zeta\ast\rho_1)\ast|\cdot|^{-1}}_{L^4(\R^3)}\\
\leq C(1+I_\Lambda)\left(\alpha^3\norm{\rho_Q-\nu}^3_\cC+\alpha^5\norm{\rho_Q-\nu}^5_\cC+\alpha^6\frac{\norm{\rho_Q-\nu}^6_\cC}{{\rm dist}(\sigma(D_Q),0)}\right)
\end{multline*}
so that
\begin{multline}
 \norm{\phi'_Q}_{L^4(\R^3)}\leq C(1+I_\Lambda)\bigg(\norm{(\nu-\rho_\gamma)\ast|\cdot|^{-1}}_{L^4(\R^3)}+\\
+\alpha^3\norm{\rho_Q-\nu}^3_\cC+\alpha^5\norm{\rho_Q-\nu}^5_\cC+\alpha^6\frac{\norm{\rho_Q-\nu}^6_\cC}{{\rm dist}(\sigma(D_Q),0)} \bigg).
\label{estim_phi_L4}
\end{multline}
As $\rho_\gamma,\nu\in \cC\cap L^1(\R^3)$, we have
$$\norm{(\nu-\rho_\gamma)\ast|\cdot|^{-1}}_{L^4(\R^3)}<\ii$$
but we do not provide a precise estimate at this point.
Now we can use the information that $\phi'_Q\in L^4(\R^3)$ to estimate $(Q_3)^{+-}$ and $(Q_3)^{-+}$, using
$$\norm{\frac{1}{D^\zeta+i\eta}\phi'_Q}_{\gS_4(\gH_\Lambda)}\leq  \frac{C}{E(\eta)^{1/4}} \norm{\phi'_Q}_{L^4(\R^3)}.$$
For any fixed $0\leq\tau<1/2$, this gives an estimate of the form
\begin{equation}
\norm{|D^\zeta|^\tau Q_3^{+-}|D^\zeta|^\tau}_{\gS_1(\gH)}\leq C\norm{\rho_Q-\nu}_\cC\norm{\phi'_Q}_{L^4(\R^3)}^2.
\end{equation}
Inserting in \eqref{decomp_rho}, we are led to
\begin{multline}
 \norm{\rho_{Q_{\rm vac}}}_{L^1(\R^3)}\leq I_\Lambda\norm{\nu-\rho_\gamma}_{L^1(\R^3)} +C(1+I_\Lambda)\norm{\rho_1}_{L^1(\R^3)}\\
+\alpha^3C(1+I_\Lambda)\norm{\rho_Q-\nu}_\cC\norm{\phi'_Q}_{L^4(\R^3)}^2
\label{estim_rho_L1}
\end{multline}
where $C$ is independent of $\Lambda$. As a conclusion, $\rho_{Q_{\rm vac}}$ hence $\rho_Q$ belong to $L^1(\R^3)$.

Let us turn to the proof of \eqref{int_rho2}. We deduce from the previous analysis that $\rho_{Q_1}\in L^1(\R^3)$ (whereas in general $Q_1\notin\gS_1(\gH_\Lambda)$) and that 
$$\int\rho_Q=\int\rho_{Q^{++}+Q^{--}}+\alpha\int\rho_{Q_1^{+-}+Q_1^{-+}}=q+\alpha\int\rho_{Q_1^{+-}+Q_1^{-+}}$$
since we know that $\gamma$, $Q_3$, $Q_5$ and $Q_6$ are all trace-class and that $\int\rho_{(K)^{+-}}=\tr(P^0_+KP^0_-)=0$ for any trace-class operator $K$. Now
$$\widehat{\rho_{Q_1}}(0)=- B^\zeta_\Lambda(0)(\widehat{\rho_Q}(0)-\widehat{\nu}(0))=- B^\zeta_\Lambda(0)(\widehat{\rho_Q}(0)-\widehat{\nu}(0))$$
which leads to
$$\int\rho_{Q_1}=-\frac{B^\zeta_\Lambda(0)}{1+\alpha B^\zeta_\Lambda(0)}(q-Z)\quad
\text{and}\quad \int(\rho_Q-\nu)=\frac{q-Z}{1+\alpha B^\zeta_\Lambda(0)}.$$
This ends the proof of Theorem \ref{thm_L1}.
\end{proof}

\begin{corollary}\label{cor_spectrum_infinite_eignv}
Let $Q$ be a minimizer for $E_{\rm r}^\nu(q)$ as in Theorem \ref{thm_L1}. If $q<Z$ (resp. $q>Z$)  then $\sigma(D_Q)$ contains an infinite sequence of eigenvalues converging to 1 (resp. to $-1$).
\end{corollary}
\begin{proof}
This is a simple adaptation of the proof of \cite[Thm A.12]{BFHS}.
\end{proof}

\subsubsection*{\bf Step 3: Properties of $q\to E_{\rm r}^\nu(q)$ and definition of $q_m$ and $q_M$.} 
As $Q\in\cQ\to\Er^\nu(Q)$ is convex, the map $q\to E_{\rm r}^\nu(q)$ is also convex. We then define 
$I=\{q\in\R\ |\ \text{(H1) holds}\}$,
where (H1) is defined in Theorem \ref{HVZ}.
Thus, for any $q\in I$, there exists a $Q\in\cQ(q)$ such that $\Er^\nu(Q)=E_{\rm r}^\nu(q)$, by Theorem \ref{HVZ}.
We introduce the following convex real functions
$f^-(q):=E_{\rm r}^\nu(q)-q$ and $f^+(q):=E_{\rm r}^\nu(q)+q$.
By \eqref{large_HVZ} and \eqref{estim_q}, $f^-$ is nonincreasing and bounded from below, $f^+$ is nondecreasing and bounded from below. Notice $\lim_{q\to\ii}f^+(q)=\ii$ and $\lim_{q\to-\ii}f^-(q)=\ii$.
Define now $q_M$ such that $f^-$ is decreasing on $(-\ii,q_M)$ and constant on $[q_M,\ii)$ (let $q_M=\ii$ if $f^-$ is decreasing), and $q_m$ such that $f^+$ is increasing on $(q_m,\ii)$ and constant on $[-\ii,q_m)$ (let $q_m=-\ii$ if $f^+$ is increasing). Remark $q_m\leq q_M$.
Next we have
\begin{eqnarray*}
q\in I & \Longleftrightarrow & \left\{\begin{array}{l}
\forall q'>q,\quad E_{\rm r}^\nu(q)<E_{\rm r}^\nu(q')+q'-q\\
\forall q'<q,\quad E_{\rm r}^\nu(q)<E_{\rm r}^\nu(q')+q-q'\\
\end{array}\right.\\
& \Longleftrightarrow & \left\{\begin{array}{l}
\forall q'>q,\quad f^+(q)<f^+(q')\\
\forall q'<q,\quad f^-(q)<f^-(q')\\
\end{array}\right.\\
& \Longleftrightarrow & \left\{\begin{array}{l}
q\in[q_m,\ii)\\
q\in(-\ii,q_M]\\
\end{array}\right.
\end{eqnarray*}
and therefore $I=[q_m,q_M]$.

\subsubsection*{\bf Step 4: The interval $[q_m,q_M]$ contains both $q_0$ and $Z$.} 
Assume now that $q_0$ satisfies $E^\nu_{\rm r}(q_0)=\min_{q\in\R}E^\nu_{\rm r}(q)$. Then 
$$E^\nu_{\rm r}(q_0)\leq E^\nu_{\rm r}(q')<E^\nu_{\rm r}(q')+|q_0-q'|$$
for any $q'\neq q_0$, ie. $q_0$ satisfies (H1). Hence $q_0\in I=[q_m,q_M]$. 

Let us now prove that $Z=\int\nu$ also belongs to $I=[q_m,q_M]$. We use classical ideas already used for the reduced Hartree-Fock theory \cite{Solovej2}. Assume first $Z>q_M$. Since $q_M\in I$, there exists a minimizer $Q_M$ in the charge sector $\cQ(q_M)$. By Theorem \ref{HVZ}, $Q_M$ satisfies the self-consistent equation 
$$Q_M+P^0_-  =  \chi_{(-\ii,\mu)}\left(D_{Q_M}\right)+\delta$$
for some $\mu\in[-1,1]$. By Corollary \ref{cor_spectrum_infinite_eignv}, $\sigma(D_{Q_M})$ contains an infinite sequence of eigenvalues converging to $1$. Since $\tr_{P^0_-}[\chi_{(-\ii,0)}\left(D_{Q_M}\right)-P^0_-]$ is known to be finite and $\delta$ is finite rank, we deduce that $\mu<1$. Hence there exists an eigenvalue $\lambda\in(\mu,1)$ of $D_{Q_M}$ with eigenfunction $\chi\in\gH_\Lambda$ which is not filled. Notice $Q_M+t|\chi\rangle\langle\chi|\in\cQ(q_M+t)$ for $t\in[0,1]$. Let us then compute, 
$$E^\nu_{\rm r}(q_M+t) \leq \Er^\nu(Q_M+t|\chi\rangle\langle\chi|) = E^\nu_{\rm r}(q_M)+t\pscal{D_{Q_M}\chi,\chi}+\frac{\alpha t^2}{2}D(|\chi|^2,|\chi|^2)$$
or equivalently
$$f^-(q_M+t)\leq f^-(q_M)+t(\lambda-1)+O(t^2)$$
which contradicts the definition of $q_M$.

Assume now $Z<q_m$ and consider a minimizer $Q_m$ for $E^\nu_{\rm r}(q_m)$. By the same arguments, it satisfies the self-consistent equation 
$$Q_m+P^0_-  =  \chi_{(-\ii,\mu')}\left(D_{Q_m}\right)+\delta$$
for some $\mu'>-1$ and the spectrum $\sigma(D_{Q_m})$ contains an infinite sequence of eigenvalues converging to -1. Thus there is an eigenvalue $ \lambda'\in(-1,\mu)$ which is completely filled, with eigenfunction $\chi'\in\gH_\Lambda$. Computing $\Er^\nu(Q_m-t|\chi'\rangle\langle\chi'|)$ and noticing $Q_m-t|\chi'\rangle\langle\chi'|\in\cQ(q_m-t)$ for any $t\in[0,1]$, we obtain
$$f^+(q_m-t)\leq f^+(q_m)-t(\lambda'+1)+O(t^2)$$
which contradicts the definition of $q_m$.

\subsubsection*{\bf Step 5: Characterization of $[q_m,q_M]$.} 
\begin{lemma}\label{str_convex_I}Assume that $q_1\neq q_2$ are such that both $E^\nu_{\rm r}(q_1)$ and $E^\nu_{\rm r}(q_2)$ admit a minimizer. Then 
\begin{equation}
\forall t\in(0,1),\qquad E^\nu_{\rm r}(tq_1+(1-t)q_2) < t\;E^\nu_{\rm r}(q_1)+(1-t)E^\nu_{\rm r}(q_2).
\label{str_convex}
\end{equation}
As a consequence, 
\begin{enumerate}
\item $[q_m,q_M]$ is the largest interval on which $q\to E^\nu_{\rm r}(q)$ is strictly convex;
\item $q_0={\rm argmin}_{I}E^\nu_{\rm r}$ is uniquely defined;
\item no minimizer exists for $E^\nu_{\rm r}(q)$ when $q$ is outside $[q_m,q_M]$;
\end{enumerate}
\end{lemma}
\begin{proof}
Assume that $Q_1$ and $Q_2$ are two minimizers of respectively $E^\nu_{\rm r}(q_1)$ and $E^\nu_{\rm r}(q_2)$, with $q_1\neq q_2$. Then by \eqref{int_rho}, $\int\rho_{Q_1}\neq\int\rho_{Q_2}$, hence $\rho_{Q_1}\neq\rho_{Q_2}$. Hence, for any $t\in(0,1)$,
\begin{multline*}
E^\nu_{\rm r}(tq_1+(1-t)q_2) \leq  \Er^\nu(tQ_1+(1-t)Q_2) < t\;\Er^\nu(Q_1)+(1-t)\Er^\nu(Q_2)\\
= t\;E^\nu_{\rm r}(q_1)+(1-t)E^\nu_{\rm r}(q_2)
\end{multline*}
where we have used the strict convexity of $f\mapsto D(f,f)$.

Inequality \eqref{str_convex} shows that $q\to E^\nu_{\rm r}(q)$ is strictly convex on $I=[q_m,q_M]$, since minimizers are known to exist for any $q\in I$. But $q\to E^\nu_{\rm r}(q)$ is linear outside $I$ and therefore $I$ is the largest interval on which $q\to E^\nu_{\rm r}(q)$ is strictly convex. The global minimizer $q_0$ of $E^\nu_{\rm r}$ on $\R$ thus on $I$ is unique.

Eventually, we prove that no minimizer exist for $E_{\rm r}^\nu(q)$ when $q\notin [q_m,q_M]$. If $q>q_M$ provides a minimizer, then since a minimizer exists for $E^\nu_{\rm r}(q_M)$, \eqref{str_convex} applied for $q_M$ and $q$ contradicts the fact that $E^\nu_{\rm r}(\cdot)$ is linear on $[q_M,\ii)$.
\end{proof}

\subsection{Proof of Theorem \ref{Ionization}}\label{sec_proof2}
If we assume $\zeta(t)=t$, the function $b_\Lambda^\zeta$ can be studied more carefully as explained in Appendix. In this case, one can prove that
$$I_\Lambda=\norm{b_\Lambda^\zeta}_{L^1(\R^3)}\leq \frac{\alpha B_\Lambda^\zeta(0)}{1-\alpha B_\Lambda^\zeta(0)}\leq\frac{2/(3\pi)\alpha \log\Lambda}{1-2/(3\pi)\alpha \log\Lambda} ,$$
when $\Lambda\geq4$ and $2/(3\pi)\alpha \log\Lambda<1$, see Proposition \ref{prop_b_Lambda_zeta_particulier}. For the sake of simplicity, we shall use the following notation in the whole proof
$$\theta:=\alpha \pi^{1/6}2^{11/6}D(\nu,\nu)^{1/2}$$
and we will assume that $\theta<1$. Later on we shall also assume that $\alpha$, $I$ and $\theta$ are small enough but we postpone this to the end of the proof and rather give precise estimates before.

\subsubsection*{\bf Step 1: \emph{A priori} estimates.}
\begin{lemma}\label{estim_apriori}
Assume that $Q\in\cQ(q)$ is a minimizer for $E_{\rm r}^\nu(q)$, for some $q\in[q_m,q_M]$. Then we have
\begin{equation}
\norm{\rho_Q-\nu}_\cC\leq \norm{\nu}_\cC.
\label{estim_unif_rho} 
\end{equation}
If moreover $\theta:=\alpha \pi^{1/6}2^{11/6}D(\nu,\nu)^{1/2}<1$, then $|D_Q|\geq 1-\theta$, hence $0\notin\sigma(D_Q)$ and, denoting $Q_{\rm vac}=\chi_{(-\ii,0)}(D_Q)-P^0_-$ we have
\begin{equation}
 \tr_{P^0_-}\left(Q_{\rm vac}\right)=0.
\end{equation}
\end{lemma}
\begin{proof}[Proof of Lemma \ref{estim_apriori}]
 We have by \eqref{estim_q}
$$\tr_{P^0_-}(D^\zeta Q)+\frac\alpha2 D(\rho_Q-\nu,\rho_Q-\nu)\leq \frac\alpha2 D(\nu,\nu)+|q|.$$
Introducing $q^+=\tr_{P^0_-}(Q^{++})\geq0$ and  $q^-=-\tr_{P^0_-}(Q^{--})\geq0$, we have
$$\tr_{P^0_-}(D^\zeta Q)=\tr(|D^\zeta|^{1/2}(Q^{++}-Q^{--})|D^\zeta|^{1/2})\geq q^++q^-\geq |q|,$$
hence \eqref{estim_unif_rho} follows.
Following \cite[p. 4495]{HLS2}, we have the operator inequality
$$|\phi'_Q|=\left|(\rho_Q-\nu)\ast\frac{1}{|\cdot|}\right|\leq \kappa\norm{\nu}_\cC|D^0|\leq \kappa\norm{\nu}_\cC|D^\zeta|$$
with $\kappa=\pi^{1/6}2^{11/6}$. Hence
\begin{equation}
|D_Q|\geq \left(1-\alpha\kappa\norm{\nu}_\cC\right)|D^\zeta|\geq 1-\theta.
\label{estim_D_Q} 
\end{equation}

The proof that $\tr_{P^0_-}(Q_{\rm vac})=0$ is the same as in \cite{HLS1,HLS2}: considering $P(t)=\chi_{(-\ii,0)}(D^\zeta+\alpha t(\rho_Q-\nu\ast|\cdot|^{-1}))$, we have by \cite[Lemma 2]{HLS1} that $\tr_{P^0_-}(P(t)-P^0_-)=\tr(P(t)-P^0_-)^3$ is for all $t\in[0,1]$ an integer which varies continuously with respect to $t$, hence, it is equal to 0 for all $t\in[0,1]$.
\end{proof}

For the rest of the proof, we work under the assumptions of Theorem \ref{Ionization}, namely we assume that $\nu$ is a radial and positive function in $L^1(\R^3)\cap \cC$ such that $\alpha \pi^{1/6}2^{11/6}D(\nu,\nu)^{1/2}<1$ and $Z=\int\nu>0$.
Let $Q$ be a minimizer for $E_{\rm r}^\nu(q)$, $q\in[q_m,q_M]$. It solves the self-consistent equation 
\begin{equation}
 Q:=Q_{\rm vac}+\gamma,\qquad Q_{\rm vac}=\chi_{(-\ii,0)}(D_Q)-P^0_-.
\label{decomposition_Q}
\end{equation}
By Lemma \ref{estim_apriori}, $\gamma$ is either $\geq0$ if $q\geq0$ or $\leq0$ if $q\leq0$. It satisfies $\norm{\gamma}_{\gS_1(\gH)}= |q|$. As $\rho_Q$ is radial by Proposition \ref{prop_unique_density_SU2}, the operator $D_Q$ is invariant under the action of $SU_2$ introduced in the proof of Proposition \ref{prop_unique_density_SU2}. In particular, we deduce that $UQ_{\rm vac}U^{-1}=Q_{\rm vac}$ for any $U\in SU_2$. Hence $\rho_{Q_{\rm vac}}$ is also a radial function. Therefore $\rho_\gamma=\rho_Q-\rho_{Q_{\rm vac}}$ is radially symmetric. 
\begin{lemma}\label{lem_estim_commut}Assume that $Q\in\cQ(q)$ is a minimizer for $E_{\rm r}^\nu(q)$ for some $q\in[q_m,q_M]$, decomposed as in \eqref{decomposition_Q}, and that $\theta:=\alpha \pi^{1/6}2^{11/6}D(\nu,\nu)^{1/2}<1$. Let be $0\leq \tau<1/2$. There exists a constant $C>0$ (depending only on $\tau$) such that
\begin{equation}
 \norm{|D^\zeta|^\tau Q_{\rm vac}}\leq \frac{C\theta}{1-\log(1-\theta)},
\label{estim_Qvac_S6}
\end{equation}
\begin{equation}
 \norm{|x|Q_{\rm vac}}\leq \frac{C}{1-\log(1-\theta)} \left(\alpha\int_{\R^3}|\rho_{Q}-\nu|+\theta\right).
\label{estim_xQvac}
\end{equation}
\end{lemma}
\begin{proof}
We have
$$Q_{\rm vac} =  \frac{\alpha}{2\pi}\int_{-\ii}^\ii \frac{1}{D^\zeta+i\eta}\phi'_Q\frac{1}{D_Q+i\eta}d\eta.$$
Hence
\begin{align*}
\norm{|D^\zeta|^{\tau}Q_{\rm vac}}_{\gS_6(\gH)} &\leq \int_{-\ii}^\ii  \frac{C\alpha\;d\eta}{\sqrt{(1-\theta)^2+\eta^2}}\norm{\frac{|D^\zeta(p)|^\tau}{(|D^\zeta(p)|^2+\eta^2)^{1/2}}\phi'_Q}_{\gS_6(\gH)}\\
 & \leq C\theta\int_{-\ii}^\ii  \frac{d\eta}{\sqrt{(1-\theta)^2+\eta^2}}\norm{\frac{|D^\zeta(p)|^\tau}{(|D^\zeta(p)|^2+\eta^2)^{1/2}}}_{L^6(\R^3)}\\
 & \leq C\theta\int_{-\ii}^\ii  \frac{d\eta}{E(\eta)^{1/2-\tau}\sqrt{(1-\theta)^2+\eta^2}}\leq\frac{C\theta}{1-\log(1-\theta)}
\end{align*}
by \eqref{estim_unif_rho}, \eqref{estim_D_Q} and the Kato-Seiler-Simon inequality \eqref{KSS}.

Notice $|x|Q_{\rm vac}=\frac{x}{|x|}\cdot xQ_{\rm vac}$ and $x|x|^{-1}$ is bounded on $\gH$. Hence for \eqref{estim_xQvac} it suffices to prove that $x_kQ_{\rm vac}$ is a bounded operator for any $k=1..3$.
We write
\begin{equation*}
 x_kQ_{\rm vac} =   \frac{\alpha}{2\pi}\int_{-\ii}^\ii\left(\left[x_k,\frac{1}{D^\zeta+i\eta}\right]\phi'_Q\frac{1}{D_Q+i\eta}+\frac{1}{D^\zeta+i\eta}x_k\phi'_Q\frac{1}{D_Q+i\eta}\right)d\eta.
\end{equation*}
Notice 
$$\left[x_k,\frac{1}{D^\zeta+i\eta}\right]=-\frac{1}{D^\zeta+i\eta}[x_k,D^\zeta+i\eta]\frac{1}{D^\zeta+i\eta}=i\frac{1}{D^\zeta+i\eta}\left(\partial_{p_k} D^\zeta\right)\frac{1}{D^\zeta+i\eta}.$$
Clearly $\partial_{p_k} D^\zeta=\alpha_k(1+|p|^2/\Lambda^2)+ 2p_k/\Lambda^2(\alp\cdot p+\beta)$, hence
$$\norm{\left(\partial_{p_k} D^\zeta\right)\frac{1}{1+|p|^2/\Lambda^2}}\leq C.$$
and by \eqref{estim_D_Q}, \eqref{estim_S6} and \eqref{estim_unif_rho}
\begin{eqnarray*}
\norm{\int_{-\ii}^\ii\left[x_k,\frac{1}{D^\zeta+i\eta}\right]\phi'_Q\frac{1}{D_Q+i\eta}d\eta} & \leq &  \int_{-\ii}^\ii\frac{d\eta}{E(\eta)\sqrt{(1-\theta)^2+\eta^2}}\norm{\phi'_Q}_{L^6(\R^3)}\\
 & \leq & \frac{C}{1-\log(1- \theta)}\norm{\nu}_\cC.
\end{eqnarray*}
Since $\rho_Q$ and $\nu$ are radial, we have by Newton's theorem
$$|x_k\phi'_Q(x)|\leq|x||\phi'_Q(x)|\leq|x|\int_{\R^3}\frac{|\rho_Q-\nu|(y)}{|x-y|}dy\leq \int_{\R^3}|\rho_Q-\nu|,$$
hence
$$\norm{\int_{-\ii}^\ii\frac{1}{D^\zeta+i\eta}x_k\phi'_Q\frac{1}{D_Q+i\eta}d\eta}\leq \frac{C}{1-\log(1- \theta)}\int_{\R^3}|\rho_Q-\nu|.$$
This ends the proof of Lemma \ref{lem_estim_commut}.
\end{proof}

\begin{lemma}\label{lem_estim_D0_gamma}We have
\begin{equation}
\norm{\gamma D^\zeta}_{\gS_1(\gH)}\leq \frac{|q|}{1-\theta}.
\label{estim_D^0_gamma}
\end{equation}
\end{lemma}
\begin{proof}
Assume for instance that $q\geq0$ and $\gamma\geq0$. By the self-consistent equation \eqref{decomposition_Q}, we have $\gamma D_{Q}\geq 0$ and $\tr(\gamma D_{Q})\leq \tr(\gamma)=q$. Hence $\norm{\gamma D_{Q}}_{\gS_1(\gH)}\leq q$. We then write
$$\gamma D_{Q}=\gamma D^\zeta\left(1+\alpha{\rm sgn}(D^\zeta)|D^\zeta|^{-1}\phi'_{Q}\right).$$
We now use that
$$\norm{|D^\zeta|^{-1}\phi'_{Q}}_{\gS_6(\gH)}\leq \kappa \norm{\rho_{Q}-\nu}_\cC\leq \kappa \norm{\nu}_\cC\leq\frac{\theta}{\alpha}$$
by \cite[p. 4495]{HLS2} and \eqref{estim_unif_rho}, so that $1+\alpha{\rm sgn}(D^\zeta)|D^\zeta|^{-1}\phi'_Q$ is invertible and 
$$\norm{\left(1+\alpha{\rm sgn}(D^\zeta)|D^\zeta|^{-1}\phi'_{Q}\right)^{-1}}\leq \frac{1}{1-\theta}.$$
This gives the result.
\end{proof}

\begin{lemma}\label{lem_estim_rho_Qvac}There exists a universal constant $C$ such that 
\begin{multline}
\left(1-C(1+I_\Lambda)^3\alpha\theta^2\right)\int|\rho_{Q_{\rm vac}}|\\
\leq (I_\Lambda+C(1+I_\Lambda)^3\alpha\theta^2)(Z+|q|)+C(1+I_\Lambda)\left(\alpha\theta^2+\frac{\theta^4}{1-\theta}\right).
\label{estim_rho_QvacM}
\end{multline}
\end{lemma}
\begin{proof}
By \eqref{estim_rho_L1}, we have
\begin{multline}
\norm{\rho_{Q_{\rm vac}}}_{L^1(\R^3)}\leq I_\Lambda(Z+|q|)+C(1+I_\Lambda)\left(\alpha\theta^2+\frac{\theta^4}{1-\theta}\right)\\
+C(1+I_\Lambda)^3\alpha^3\norm{\phi'_Q}_{L^4(\R^3)}^2\norm{\nu}_\cC.
\label{estim_rho_prime}
\end{multline}
Notice for any $\rho$
\begin{equation}
\norm{\rho\ast|\cdot|^{-1}}_{L^4(\R^3)}\leq C\norm{\rho}_\cC^{1/2}\norm{\rho}_{L^1(\R^3)}^{1/2}
\label{estim_nu_L4} 
\end{equation}
which is proved by writing
\begin{align*}
\norm{\rho\ast|\cdot|^{-1}}_{L^4(\R^3)}&\leq C\norm{\widehat{\rho}(k)|k|^{-2}}_{L^{4/3}(\R^3)}\\
&\leq \norm{\rho}_{L^1(\R^3)}\norm{|k|^{-2}}_{L^{4/3}(B(0,r))}+\norm{\rho}_\cC\norm{|k|^{-1}}_{L^{4}(\R^3\setminus B(0,r))} 
\end{align*}
and optimizing in $r$. Using \eqref{estim_nu_L4} for $\rho=\rho_Q-\nu$ and \eqref{estim_unif_rho}, we get
$$\norm{\phi'_Q}_{L^4(\R^3)}^2\leq C\norm{\nu}_\cC\int|\rho_Q|+\nu\leq C\norm{\nu}_\cC\left(Z+|q|+\int|\rho_{Q_{\rm vac}}|\right).$$
Inserting this in \eqref{estim_rho_prime} yields the result.
\end{proof}

\subsubsection*{\bf Step 2: Lieb's argument.} We now use ideas from Lieb \cite{Lieb} to obtain a bound on $q_M$. We denote by $Q$ a minimizer for $E^\nu(q_M)$ which exists by Theorem \ref{exists}. As $q_M\geq Z>0$, we can decompose $Q$ as in \eqref{decomposition_Q}:
\begin{equation}
Q=\chi_{(-\ii,\mu)}(D_Q)-P^0_-+\delta=Q_{\rm vac}+\gamma
\label{scf_eq_Q_ionization}
\end{equation}
with $\gamma\geq0$.
Using that $(D_Q-1)\gamma\leq0$ due to \eqref{scf_eq_Q_ionization}, we infer
\begin{equation}
0\geq \tr(|x|(D_Q-1)\gamma) =  \tr(|x|(D^\zeta-1)\gamma)+\alpha\int_{\R^3}|x|\phi'_Q(x)\rho_\gamma(x)dx.
\label{estim_Lieb}
\end{equation}

\begin{lemma}\label{lem_estim_premier_terme}
There exists a universal constant $C$ such that
\begin{equation}
\tr(|x|(D^\zeta-1)\gamma)\geq -\frac{Cq_M}{1-\theta}\left(\frac1\Lambda+\frac{\alpha q_M+\alpha Z+\theta}{1-\log (1-\theta)} \right).
\label{estim_premier_terme}
\end{equation}
\end{lemma}

The proof of Lemma \ref{lem_estim_premier_terme} will be given at the end of this section. Now we assume that $\alpha$, $I_\Lambda$ and $\theta$ are all small enough. Then \eqref{estim_premier_terme} becomes
\begin{equation}
\tr(|x|(D^\zeta-1)\gamma)\geq -Cq_M\left(\frac1\Lambda+\alpha q_M+\alpha Z+\theta\right).
\label{estim_premier_terme_bis}
\end{equation}
and \eqref{estim_rho_QvacM} becomes
\begin{equation}
\int|\rho_{Q_{\rm vac}}|\leq C(I_\Lambda+\alpha\theta^2)(Z+q_M)+C\left(\alpha\theta^2+\theta^4\right).
\label{estim_rho_QvacM_bis}
\end{equation}
To estimate the second term in \eqref{estim_Lieb}, we write
\begin{multline*}
\int_{\R^3}|x|\phi'_Q(x)\rho_\gamma(x)dx =\iint_{\R^6}\frac{|x|+|y|}{2|x-y|}\rho_\gamma(x)\rho_\gamma(y)dx\,dy \\
  +\iint_{\R^6}\frac{|x|(\rho_{Q_{\rm vac}}-\nu)(y)\rho_\gamma(x)}{|x-y|}dx\,dy
\end{multline*}
and notice
\begin{equation}
 \iint_{\R^6}\frac{|x|+|y|}{2|x-y|}\rho_\gamma(x)\rho_\gamma(y)dx\,dy\geq \frac{q_M^2}{2}
\end{equation}
since $\rho_\gamma\geq0$ and $|x-y|\leq|x|+|y|$.
Using Newton's theorem we infer
\begin{equation}
 \iint_{\R^6}\frac{|x|\nu(y)\rho_\gamma(x)}{|x-y|}dx\,dy\leq Zq_M,
\label{estim_interaction}
\end{equation}
\begin{equation*}
\iint_{\R^6}\frac{\left|\rho_{Q_{\rm vac}}(y)\right||x|\rho_\gamma(x)}{|x-y|}dx\,dy \leq C(I_\Lambda+\alpha\theta^2)q_M^2+Cq_M(I_\Lambda+\alpha\theta^2)Z+Cq_M\left(\alpha\theta^2+\theta^4\right)
\end{equation*}
by \eqref{estim_rho_QvacM_bis} and since both $\nu$, $\rho_\gamma$ and $\rho_{Q_{\rm vac}}$ are radial functions.
Collecting estimates and using that $\alpha$, $I$ and $\theta$ are small enough, we obtain the following estimate
\begin{equation}
(1-C\alpha\log\Lambda)q_M\leq 2(1+C\alpha\log\Lambda)Z+\frac{C}{\Lambda}+C\theta.
\end{equation}
The proof for $q_m$ is the same, using that in this case $\gamma\leq0$ and instead of \eqref{estim_interaction}
$$-\iint_{\R^6}\frac{|x|\nu(y)\rho_{\gamma}(x)}{|x-y|}dx\,dy\geq0.$$

\begin{proof}[Proof of Lemma \ref{lem_estim_premier_terme}]
 For the second term of \eqref{estim_Lieb}, we compute
\begin{align}
 \tr(|x|(D^\zeta-1)\gamma) & = \tr(|x|(|D^\zeta|-1)\gamma)-2\tr(|x|D^\zeta P^0_-\gamma)\label{decomp_kinetic}\\
 & =\tr(|x|(|D^\zeta|-1)\gamma)+2\tr(|x|D^\zeta Q_{\rm vac}\gamma)\nonumber\\
 & =\tr(|x|(|D^\zeta|-1)\gamma)+2\tr([|x|,D^\zeta] Q_{\rm vac}\gamma)+2\tr(|x|Q_{\rm vac}\gamma D^\zeta)\nonumber
\end{align}
where we have used that $\chi_{(-\ii,0]}(D_Q)\gamma=0$ by \eqref{scf_eq_Q_ionization}. 
One computes
\begin{multline*}
(|D^\zeta(p)|-1)|x|+|x|(|D^\zeta(p)|-1)
 = (E(p)-1)|x|+|x|(E(p)-1)\\
+ \sum_{k=1}^3\frac{E(p)p_k}{\Lambda^2}[p_k,|x|]+[|x|,p_k]\frac{p_kE(p)}{\Lambda^2}
+\frac{1}{\Lambda^2}\sum_{k=1}^3p_k(E(p)|x|+|x|E(p))p_k.
\end{multline*}
Next we use a result of Lieb \cite{Lieb} which says that
$$(E(p)-1)|x|+|x|(E(p)-1)\geq0.$$
We obtain
\begin{equation*}
(|D^\zeta(p)|-1)|x|+|x|(|D^\zeta(p)|-1)
\geq \frac{i}{\Lambda^2}\sum_{k=1}^3\left[\frac{x_k}{|x|},E(p)p_k\right]
\end{equation*}
and
\begin{equation*}
 \tr(|x|(|D^\zeta|-1)\gamma)\geq -\frac{1}{\Lambda^2}\sum_{k=1}^3\norm{\frac{E(p)p_k}{|D^\zeta(p)|}}\norm{\gamma|D^\zeta|^{-1}}_{\gS_1(\gH)}.
\end{equation*}
Notice
$$\frac{E(p)p_k}{|D^\zeta(p)|}=\frac{p_k}{1+|p|^2/\Lambda^2}\leq \frac\Lambda2$$
hence, using \eqref{estim_D^0_gamma}, we obtain
\begin{equation}
 \tr(|x|(|D^\zeta|-1)\gamma)\geq -\frac{C}{(1-\theta)\Lambda}q_M.
\label{estim_below_almost_positive}
\end{equation}
Let us now estimate the last term of the r.h.s. of \eqref{decomp_kinetic}. Using \eqref{estim_xQvac} and \eqref{estim_D^0_gamma}, we obtain the following estimate:
\begin{equation}
 \left|\tr(|x|Q_{\rm vac}\gamma D^\zeta)\right|\leq \frac{C}{(1-\theta)(1-\log(1- \theta))} \left(\alpha q_M+\alpha Z+\theta\right)q_M.
\label{estimate1}
\end{equation}
Eventually we estimate the second term of the r.h.s. of \eqref{decomp_kinetic}. We compute
\begin{align*}
\left[D^\zeta,|x|\right] & = [\alp\cdot p,|x|]+\frac{\beta}{\Lambda^2} [|p|^2,|x|]+\sum_{k=1}^3\frac{\alpha_k}{\Lambda^2} [p_k|p|^2,|x|]\\
& = -i\sum_{k=1}^3\alpha_k\frac{x_k}{|x|}+\frac{\beta}{\Lambda^2} [|p|^2,|x|]+\sum_{k=1}^3\frac{\alpha_k}{\Lambda^2} [ p_k|p|^2,|x|],
\end{align*}
$$[|p|^2,|x|]=\frac{2}{|x|}-2ip\cdot\frac{x}{|x|},$$
$$[ p_k|p|^2,|x|]=\frac{2}{|x|}p_k+2p_k\frac{1}{|x|}-2p\cdot x\frac{x_k}{|x|^3}+2ip_kp\cdot\frac{x}{|x|}-i|p|^2\frac{x_k}{|x|}.$$
Hence, using Hardy's inequality which tell us that $|p|^{-1}|x|^{-1}$ is a bounded operator on $\gH$, we easily deduce that
$$\left[|x|,D^\zeta\right]=-\frac{2}{\Lambda^2}\frac{1}{|x|}\alp\cdot p+A$$
where $A$ is an operator satisfying
$\norm{|D^\zeta|^{-1}A}\leq C$
for a universal constant $C$ independent of $\Lambda$. Next we write
\begin{align}
\tr([|x|,D^\zeta] Q_{\rm vac}\gamma)  &= -\frac{2}{\Lambda^2}\tr\left(\frac{1}{|p|}\frac{1}{|x|}\alp\cdot pQ_{\rm vac}\gamma|p|\right)+ \tr\left((D^\zeta)^{-1}A Q_{\rm vac}\gamma D^\zeta\right)\nonumber\\
&\geq -\frac{C}{\Lambda^{4/3}}\norm{\frac{|p|}{\Lambda^{2/3}}Q_{\rm vac}}\frac{q_M}{1-\theta}-C\norm{Q_{\rm vac}}\frac{q_M}{1-\theta}\nonumber\\
&\geq -C\frac{q_M\theta}{(1-\theta)\log(1-\theta)}
\label{estimate2}
\end{align}
by Lemma \ref{lem_estim_D0_gamma} and Lemma \ref{lem_estim_commut} with $\tau=0$ and $\tau=1/3$.
Inserting \eqref{estim_below_almost_positive}, \eqref{estimate1} and \eqref{estimate2} in Formula \eqref{decomp_kinetic}, we obtain \eqref{estim_premier_terme}. 
This ends the proof of Lemma \ref{lem_estim_premier_terme}.
\end{proof}

\appendix

\section{Study of the function $b^\zeta_\Lambda$}
This appendix is devoted to the decay properties of $b_\Lambda^\zeta$, for the different cut-offs chosen in this article. The function $b_\Lambda^\zeta$ plays an important role in the model as it can be interpreted as the linear response of the vacuum in the presence of an external field, as shown by Formula \eqref{decomp_rho}.
We recall that
\begin{equation}\label{def_b_Lambda2}
\widehat{b_\Lambda^\zeta}(k)=(2\pi)^{3/2}\frac{\alpha B^\zeta_\Lambda(k)}{1+\alpha B^\zeta_\Lambda(k)}.
\end{equation}

\subsection{Study of $b^0_\Lambda$ when $\gH=\gH_\Lambda$ and $\zeta=0$}\label{sec_b_0_sharp}
We start with the sharp cut-off case $\zeta=0$ and $\gH=\gH_\Lambda$. In this case
\begin{equation}
B_\Lambda^0(k)=-\frac{1}{\pi^2 |k|^2}\int_{\substack{|q+k/2|\leq\Lambda,\\ |q-k/2|\leq\Lambda}} 
\frac{(q+k/2)\cdot(q-k/2)+1-E(q+k/2)E(q-k/2)}{E(q+k/2)E(q-k/2)(E(q+k/2)+E(q-k/2))}dq
\label{def_B_Lambda_bis}
\end{equation}
is defined for $|k|\leq 2\Lambda$. Following \cite{PauliRose}, for any $q\in\R^3$ we introduce as new variables the azimuth angle $\phi$ around the axis parallel to $k$ and
\begin{equation}
\left\{\begin{array}{rcl}
 v & = & (E(q+k/2)-E(q-k/2))/2,\\
 w & = & (E(q+k/2)+E(q-k/2))/2.
\end{array}\right. 
\label{change_variables}
\end{equation}
Then integrating over $\{q\in\R^3\ |\ |q+k/2|\leq\Lambda,\ |q-k/2|\leq\Lambda\}$ is easily shown to be equivalent to integrate over the new variables $(u,v,\phi)\in\R\times\R\times[0,2\pi)$ with the three conditions
\begin{equation}
1\leq v+w\leq \sqrt{1+\Lambda^2},\qquad  1\leq w-v\leq \sqrt{1+\Lambda^2},
\label{cond_v_w}
\end{equation}
\begin{equation}
\sqrt{1+|k|^2/4}\leq w\leq \sqrt{1+\Lambda^2},\qquad |v|\leq \frac{|k|}2\sqrt{\frac{w^2-|k|^2/4-1}{w^2-|k|^2/4}}.
\label{equiv_ineg_triang}
\end{equation}
Eventually \eqref{cond_v_w} and \eqref{equiv_ineg_triang} are equivalent to 
\begin{equation}
 \sqrt{1+|k|^2/4}\leq w\leq \sqrt{1+\Lambda^2}
\label{cond_w}
\end{equation}
\begin{equation}
 |v|\leq \min\left(w-1,\ \sqrt{1+\Lambda^2}-w,\ \frac{|k|}2\sqrt{\frac{w^2-|k|^2/4-1}{w^2-|k|^2/4}}\right).
\label{cond_v}
\end{equation}
An explicit computation shows that 
\begin{multline}
\min\left(w-1,\ \sqrt{1+\Lambda^2}-w,\ \frac{|k|}2\sqrt{\frac{w^2-|k|^2/4-1}{w^2-|k|^2/4}}\right) \\ =\left\{\begin{array}{ll}
 \frac{|k|}2\sqrt{\frac{w^2-|k|^2/4-1}{w^2-|k|^2/4}}& \text{when } w\leq W_\Lambda(|k|)\\
\sqrt{1+\Lambda^2}-w& \text{when } w\geq W_\Lambda(|k|)\\
\end{array}\right. 
\end{multline}
where $W_\Lambda(r):=(\sqrt{1+\Lambda^2}+\sqrt{1+(\Lambda-r)^2})/2$
is the unique root of the fourth order polynomial equation 
$$\sqrt{1+\Lambda^2}-w= \frac{|k|}2\sqrt{\frac{w^2-|k|^2/4-1}{w^2-|k|^2/4}}$$
in $[\sqrt{1+|k|^2/4},\sqrt{1+\Lambda^2}]$.
Inserting this in the definition of $B_\Lambda(k)$ and using that $dq=(2/|k|)E(q+k/2)E(q-k/2)dvdwd\phi$, see \cite[Eq. (12)]{PauliRose}, we find
\begin{multline}
B_\Lambda^0(k)  = - \frac{8}{\pi|k|^3} \Bigg(\int_{\sqrt{1+|k|^2/4}}^{W_\Lambda(|k|)}dw\int_0^{\frac{|k|}2\sqrt{\frac{w^2-|k|^2/4-1}{w^2-|k|^2/4}}}dv\frac{v^2-|k|^2/4}{w}\\
+\int_{W_\Lambda(|k|)}^{\sqrt{1+\Lambda^2}}dw\int_0^{\sqrt{1+\Lambda^2}-w}dv\frac{v^2-|k|^2/4}{w}\Bigg).
\end{multline}
Letting $z=\sqrt{\frac{w^2-|k|^2/4-1}{w^2-|k|^2/4}}$ in the first integral and $z=2(\sqrt{1+\Lambda^2}-w)/|k|$ in the second, we obtain
\begin{multline}
B_\Lambda^0(k)  = \frac1\pi\int_{0}^{Z_\Lambda(|k|)}\frac{z^2-z^4/3}{(1-z^2)(1+|k|^2(1-z^2)/4)}dz\\
 +\frac{|k|}{2\pi}\int_{0}^{Z_\Lambda(|k|)}\frac{z-z^3/3}{\sqrt{1+\Lambda^2}-|k|z/2}dz\label{expression_B}
\end{multline}
where we have defined
$$Z_\Lambda(r)=\sqrt{\frac{W_\Lambda(r)^2-r^2/4-1}{W_\Lambda(r)^2-r^2/4}}=\frac{\sqrt{1+\Lambda^2}-\sqrt{1+(\Lambda-r)^2}}{r}.$$
The first term of \eqref{expression_B} was already present in \cite{PauliRose}, whereas the second term was ignored by Pauli and Rose. 
An explicit computation of the integrals in \eqref{expression_B} yields
\begin{multline}
B_\Lambda^0(k)  = \frac{1}{\pi|k|^3}\Bigg\{ -\frac{4rZ_\Lambda(r)}{3}-\frac23(r^2-2)\sqrt{4+r^2}\text{arctanh}\left(\frac{rZ_\Lambda(r)}{\sqrt{4+r^2}}\right)\\
+\frac{r^3}3\log\left(\frac{1+Z_\Lambda(r)}{1-Z_\Lambda(r)}\right)
+\frac{8E(\Lambda)^3}{3}\left(1-\frac{3r^2}{4E(\Lambda)^2}\right)\log\frac{Z_\Lambda(r)}{E(\Lambda)}+\frac{44}9E(\Lambda)^3\\
-2E(\Lambda)r^2-2\left(3E(\Lambda)^2+1+r^2\right)Z_\Lambda(r)+3E(\Lambda)Z_\Lambda(r)^2-\frac{8}{9}Z_\Lambda(r)^3\Bigg\}.
\label{expression_B2}
\end{multline}
To avoid any further notation, we now see $B_\Lambda^0$ as a function of $|k|$.
Using Formula \eqref{expression_B2}, one can prove the
\begin{prop}[Regularity of $B_\Lambda^0$]
 \label{smoothness_B_Lambda_0}
Let be $\Lambda > 0$. The function $r\mapsto B^0_\Lambda(r)$ extends to a non-negative, $C^1$ function on $\R_+$, which vanishes on $[2 \Lambda, + \infty)$. Moreover, it is of class $C^3$ on $[0, 2 \Lambda]$. Eventually, we have
$$\boB^0_\Lambda(0) = B_\Lambda^0 = \frac{2}{3 \pi} \ln(\Lambda)  + O(1),\qquad \frac{d \boB^0_\Lambda}{dr}(0) = -\frac{1}{8 \pi \Lambda} + \underset{\Lambda \to + \infty}{O} \Big( \frac{1}{\Lambda^3} \Big),$$
$$\frac{d^2\boB^0_\Lambda}{dr^2}(0) = - \frac{2}{15 \pi} + \underset{\Lambda \to + \infty}{O} \Big( \frac{1}{\Lambda^2} \Big),\qquad\frac{d^3\boB^0_\Lambda}{dr^3}(0) = \frac{3}{4 \pi \Lambda} + \underset{\Lambda \to + \infty}{O} \Big( \frac{1}{\Lambda^2} \Big)$$
$$\boB^0_\Lambda(2 \Lambda) = \frac{d\boB^0_\Lambda}{dr}(2 \Lambda) = 0,\quad
\frac{d^2\boB^0_\Lambda}{dr^2}(2 \Lambda) = \frac{\Lambda}{4 \pi E(\Lambda)^3},\quad \frac{d^3\boB^0_\Lambda}{dr^3}(2 \Lambda) = \frac{5\Lambda^2-1}{8 \pi E(\Lambda)^5}.$$
\end{prop}

By Proposition \ref{smoothness_B_Lambda_0}, $\boB^0_\Lambda$ is a non-negative, continuous function with compact support. Therefore, by \eqref{def_b_Lambda2}, $b^0_\Lambda$ is a smooth function which reads (using the inverse Fourier formula for spherically symmetric functions),
\begin{equation}
\forall x \in \R^3 \setminus \{ 0 \},\quad  b^0_\Lambda(x) = \frac{2 \pi}{|x|} \int_0^{2 \Lambda} \frac{\alpha \boB^0_\Lambda(r)}{1 + \alpha \boB^0_\Lambda(r)} \sin(r |x|) r dr. 
\label{formula_Fourier_inverse_radial}
\end{equation}
In particular we get the bound
\begin{equation}
 |b^0_\Lambda(x)|\leq \frac{16\pi \Lambda^3}{3},
\label{bound_unif_b_Lambda_0}
\end{equation}
which shows that $b_\Lambda^0\in L^\ii(\R^3)$.
This becomes after three integrations by parts,
\begin{align}
b^0_\Lambda(x) & = \frac{2\pi}{ |x|^4} \Bigg\{2\alpha\Lambda (B_\Lambda^0)''(2\Lambda)\cos(2\Lambda|x|)+\frac{2\alpha (B_\Lambda^0)'(0)}{(1+\alpha B_\Lambda^0(0))^2} \nonumber \\ 
&\qquad  - \int_0^{2 \Lambda} \Bigg( \frac{\alpha r (\boB^0_\Lambda)^{(3)}(r)}{(1 + \alpha \boB^0_\Lambda(r))^2} - \frac{6 \alpha^2 r (\boB^0_\Lambda)'(r)( \boB^0_\Lambda)''(r)}{(1 + \alpha \boB^0_\Lambda(r))^3} + \frac{6 \alpha^3 r (\boB^0_\Lambda)'(r)^3}{(1 + \alpha \boB^0_\Lambda(r))^4} \nonumber\\ 
& \qquad + \frac{3 \alpha (\boB^0_\Lambda)''(r)}{(1 + \alpha \boB^0_\Lambda(r))^2} - \frac{6 \alpha^2 (\boB^0_\Lambda)'(r)^2}{(1 + \alpha \boB^0_\Lambda(r))^3} \Bigg) \cos(r |x|) dr \Bigg\},\label{formula_b_radial_x4}
\end{align}
which yields by Proposition \ref{smoothness_B_Lambda_0}
\begin{equation}\label{b_decay}
|b_\Lambda(x)| \leq \frac{C_{\alpha,\Lambda}}{|x|^4},
\end{equation}
for some constant $C_{\alpha,\Lambda}$ depending on $\alpha$ and $\Lambda$.
With \eqref{bound_unif_b_Lambda_0}, this proves the
\begin{prop}\label{prop_b_Lambda_0}Assume $\gH=\gH_\Lambda$ and $\zeta=0$. Let be $\alpha\geq 0$ and $\Lambda>0$. 
Then $b_\Lambda^0$ belongs to $L^1(\R^3)$.
\end{prop}

\begin{remark}\rm 
It can be seen that $\boB_\Lambda^0(\Lambda r)\to \boB^0_\ii(r)$ where
$$\boB^0_\ii(r)=\left\{\begin{array}{r}
\frac{4}{\pi r^3}\Big[\left(\frac23 -\frac{r^2}2\right) \log\left(\frac{r}{2}\right)-\frac{(r-1)^3}{36}+\frac{(r-1)^2}{6}+\left(\frac{r^2}{4}-\frac{7}{12}\right)(r-1)\\
-\frac{r^2}{4}+\frac49\Big]\qquad \text{when } 1\leq r\leq 2,\\
\frac{4}{\pi r^3}\Big[\left(\frac23 -\frac{r^2}2\right) \log\left(\frac{2-r}{2}\right)-\frac{(1-r)^3}{36}+\frac{(1-r)^2}{6}+\left(\frac{r^2}{4}-\frac{7}{12}\right)(1-r)\\
-\frac{r^2}{4}+\frac49+\frac{r^3}{6}\log\left(\frac{2-r}{r}\right)\Big]\qquad \text{when } 0< r\leq 1.
\end{array}\right.$$
The convergence holds in $C^2([0,2))$. Notice the two terms appearing \eqref{expression_B} separately converge to a function which is not differentiable at $r=1$, but there is some cancellation occurring. It can also be proved that $\norm{b_\Lambda^0}_{L^1(\R^3)}$ is indeed uniformly bounded independently of the cut-off $\Lambda$, but we do not need that in this article.
\end{remark}

\subsection{Study of $b^\zeta_\Lambda$ when $\gH=L^2(\R^3,\C^4)$ and $\zeta\neq 0$}
When $\gH=L^2(\R^3,\C^4)$ and $\zeta\neq 0$ satisfies \eqref{prop_zeta_1}--\eqref{prop_zeta_3}, the same changes of variables, followed by $t=\sqrt{1-u^2}$, lead to 
\begin{equation}
 B_\Lambda^\zeta(rk)  = \pi^{-1}\int_{0}^{1}\frac{dt}{t(1+|k|^2t^2/4)}\int_0^{\sqrt{1-t^2}}\frac{1-u^2}{1+\Psi(|k|,t,u)}du,
\label{formula_B_zeta_general}
\end{equation}
where
$$\Psi(|k|,t,u)=\frac12(\eta(w+v)+\eta(w-v))+\frac{v}{2w}(\eta(w+v)-\eta(w-v)),$$
$$\eta(x)=\zeta\left(\frac{x^2-1}{\Lambda^2}\right),\quad v=\frac{|k|u}{2}\quad\text{and}\quad w=\sqrt{\frac{|k|^2}{4}+\frac1{t^2}}.$$

\begin{prop}\label{prop_b_Lambda_zeta}Assume $\gH=L^2(\R^3,\C^4)$ and that $\zeta\neq0$ satisfies \eqref{prop_zeta_1}--\eqref{prop_zeta_3}. Let be $\alpha\geq 0$ and $\Lambda>0$. 
The function $b_\Lambda^\zeta$ belongs to $L^1(\R^3)$.
\end{prop}
\begin{proof}
As before, we consider $B_\Lambda^\zeta$ as a function of $|k|$ to simplify the notation. We shall prove an estimate of the form
\begin{equation}
|B_\Lambda^\zeta(r)|\leq \frac{C_\Lambda}{1+r^{2\epsilon}},\qquad \left|\frac{d}{dr}B_\Lambda^\zeta(r)\right|\leq \frac{C_\Lambda}{1+r^{1+2\epsilon}},
\label{estim_B_derivees}
\end{equation}
\begin{equation}
\left|\frac{d^2}{dr^2}B_\Lambda^\zeta(r)\right|\leq \frac{C_\Lambda\log(2+r)}{1+r^{2+2\epsilon}},\qquad
 \left|\frac{d^3}{dr^3}B_\Lambda^\zeta(r)\right|\leq \frac{C_\Lambda}{1+r^{2+2\epsilon}}
\label{estim_B_derivees2}
\end{equation}
where $\epsilon>0$ is given by \eqref{prop_zeta_2}. The result will follow using a formula similar to \eqref{formula_b_radial_x4}. We first notice that $v\leq w$, so that $\eta(w+v)-\eta(w-v)\geq0$ as $\zeta$ is nonincreasing. Hence, by \eqref{prop_zeta_2},
\begin{equation}
1+\Psi\geq 1+ \frac{\eta(w)}2\geq c\left|\frac{r^2}{4}+\frac{1}{t^2}\right|^\epsilon=\frac{c\left|\frac{r^2t^2}{4}+1\right|^\epsilon}{t^{2\epsilon}}.
\label{estim_Psi_below} 
\end{equation}
Inserting in \eqref{formula_B_zeta_general}, we obtain
\begin{equation}
|B_\Lambda^\zeta(r)|\leq C\int_0^1\frac{dt}{t^{1-2\epsilon}\left(\frac{r^2t^2}{4}+1\right)^{1+\epsilon}}\leq \frac{C}{1+r^{2\epsilon}}.
\end{equation}
For the three first derivatives of $B_\Lambda^\zeta$, we invoke the following
\begin{lemma}\label{estim_Psi}
We have for any $p=1,2,3$
\begin{equation}
\left|\frac{\partial^p}{\partial r^p}\left(\frac1{1+\Psi(r,t,u)}\right)\right|\leq \frac{C}{(1+\Psi(r,t,u))\left(1+(1-u)^pr^p\right)}.
\end{equation}
\end{lemma}
Assuming Lemma \ref{estim_Psi} holds, we can write 
\begin{align*}
(B_\Lambda^\zeta)'(r)& =-(2\pi)^{-1}\int_{0}^{1}\frac{rt\,dt}{(1+r^2t^2/4)^2}\int_0^{\sqrt{1-t^2}}\frac{1-u^2}{1+\Psi(r,t,u)}du\\
& +\pi^{-1}\int_{0}^{1}\frac{dt}{t(1+r^2t^2/4)}\int_0^{\sqrt{1-t^2}}(1-u^2)\frac{\partial}{\partial r}\left(\frac{1}{1+\Psi(r,t,u)}\right)du,
\end{align*}
hence
\begin{align*}
|(B_\Lambda^\zeta)'(r)|&\leq C\int_{0}^{1}\frac{rt^{1+2\epsilon}dt}{\left(\frac{r^2t^2}{4}+1\right)^{2+\epsilon}}+C\int_0^1\frac{t^{2\epsilon-1}dt}{\left(\frac{r^2t^2}{4}+1\right)^{1+\epsilon}}\int_0^1\frac{(1-u)du}{1+(1-u)r}\\
&\leq \frac{C}{1+r^{1+2\epsilon}}.
\end{align*}
The proof of \eqref{estim_B_derivees2} is similar. Therefore, we omit it. Instead we turn to the
\begin{proof}[Proof of Lemma \ref{estim_Psi}]
We have
$$\frac{\partial}{\partial r}\left(\frac1{1+\Psi(r,t,u)}\right)=-\frac{\frac{\partial}{\partial r}\Psi(r,t,u)}{(1+\Psi(r,t,u))^2},$$
so that we have to prove that
$$\left|\frac{\frac{\partial}{\partial r}\Psi(r,t,u)}{1+\Psi(r,t,u)}\right|\leq \frac{C}{1+(1-u)r}.$$
Since $\Psi(r,t,u)=\frac12(\eta(w+v)+\eta(w-v))+w'u(\eta(w+v)-\eta(w-v))$, we obtain
\begin{multline}
\frac{\partial}{\partial r}\Psi(r,t,u)=\frac12\big\{(w'+v')\eta'(w+v)+(w'-v')\eta'(w-v)\big\}+w''u\big\{\eta(w+v)-\eta(w-v)\big\}\\
 + w'u\big\{(w'+v')\eta'(w+v)-(w'-v')\eta'(w-v)\big\}.
\label{1_derivee_Psi}
\end{multline}
Next we remark that $v'=u\leq 1$, $v''=0$, $w'\leq 1/2$, $w''\leq C(1+r)^{-1}$ and $w'''\leq C(1+r)^{-2}$. Using $1+\Psi(r,t,u)\geq 1+\eta(w+v)+\eta(w-v)$, and Assumption \eqref{prop_zeta_3}, we obtain an estimate of the form
$$\frac{\left|\frac{\partial}{\partial r}\Psi(r,t,u)\right|}{1+\Psi(r,t,u)}\leq C\left(\frac{1}{1+r}+\frac{1}{1+|w+v|}+\frac{1}{1+|w-v|}\right).$$
Eventually, we use
$$|w+v|=w+v\geq w\geq \frac{r}{2},\quad\text{ and }\quad |w-v|=w-v\geq \frac{r}{2}(1-u).$$
The proof for the other derivatives is similar.
\end{proof}
\end{proof}

We end this section with
\begin{prop}Assume $\gH=L^2(\R^3,\C^4)$ and that $\zeta\neq0$ satisfies \eqref{prop_zeta_1}--\eqref{prop_zeta_3}.
We have, as $\Lambda\to\ii$,
$$B_\Lambda^\zeta(0)=\frac{2}{3\pi}\log\Lambda+O(1).$$
\end{prop}
\begin{proof}
Taking $z=\sqrt{1-t^2}$, we obtain
\begin{equation*}
B_\Lambda^\zeta(0) =\frac1\pi\int_0^{1}\frac{z^2-z^4/3}{(1-z^2)\left(1+\zeta\left(\frac{z^2}{\Lambda^2(1-z^2)}\right)\right)}dz.
\end{equation*}
As $\zeta\geq0$, we have
$$\frac1\pi\int_0^{\frac{\Lambda}{E(\Lambda)}}\frac{z^2-z^4/3}{(1-z^2)\left(1+\zeta\left(\frac{z^2}{\Lambda^2(1-z^2)}\right)\right)}dz\leq \frac1\pi\int_0^{\frac{\Lambda}{E(\Lambda)}}\frac{z^2-z^4/3}{1-z^2}dz=B_\Lambda^0(0).$$
To get a lower bound, we use that $\zeta$ is smooth, and write that $\zeta(x)\leq cx$ for any $0\leq x\leq 1$ and some $c>0$. We obtain
\begin{multline*}
\frac1\pi\int_0^{\frac{\Lambda}{E(\Lambda)}}\frac{z^2-z^4/3}{(1-z^2)\left(1+\zeta\left(\frac{z^2}{\Lambda^2(1-z^2)}\right)\right)}dz
\geq \frac1\pi\int_0^{\frac{\Lambda}{E(\Lambda)}}\frac{z^2-z^4/3}{(1-z^2)\left(1+c\frac{z^2}{\Lambda^2(1-z^2)}\right)}dz \\
\geq \frac1{\pi\sqrt{1-c/\Lambda^2}}\int_0^{\frac{\Lambda\sqrt{1-c/\Lambda^2}}{E(\Lambda)}}\frac{z^2-z^4/3}{1-z^2}dz=B_\Lambda^0(\Lambda)-\frac{\log(1+c)}{3\pi}+o(1),
\end{multline*}
as $\zeta(x)\leq cx$ for any $0\leq x\leq 1$ and some $c>0$. Finally, by \eqref{prop_zeta_2},
\begin{multline*}
\left|\int_{\frac{\Lambda}{E(\Lambda)}}^{1}\frac{(z^2-z^4/3)dz}{(1-z^2)\left(1+\zeta\left(\frac{z^2}{\Lambda^2(1-z^2)}\right)\right)}\right|\leq CE(\Lambda)^\epsilon \int_{\frac{\Lambda}{E(\Lambda)}}^{1}\frac{dz}{(1-z^2)^{1-\epsilon/2}}=O(1).
\end{multline*}
This yields the result.
\end{proof}

\subsection{Study of $b^\zeta_\Lambda$ for $\zeta(t)=t$}
We finally turn to the special cut-off $\zeta(t)=t$ which was used in the study of the ionization in Theorem \ref{Ionization}. Formula \eqref{formula_B_zeta_general} yields in this case
\begin{equation}
B_\Lambda^\zeta(k)  = \pi^{-1}\int_{0}^{1}\frac{dt}{t(1+|k|^2t^2/4)}\int_0^{\sqrt{1-t^2}}\frac{1-u^2}{1+\frac1{\Lambda^2}\left(\frac{|k|^2}4-1+\frac1{t^2}\right)+\frac{3|k|^2}{4\Lambda^2}u^2}du.
\label{formula_B_zeta_particulier}
\end{equation}
Notice that $B_\Lambda^\zeta$ is nonnegative, when $\Lambda>1/\sqrt{2}$, and
\begin{align*}
 B_\Lambda^\zeta(0)  &= \pi^{-1}\int_{0}^{1}\frac{zdz}{1-z^2}\int_0^{z}\frac{1-u^2}{1+\frac{z^2}{\Lambda^2(1-z^2)}}du =\pi^{-1}\int_{0}^{1}\frac{z^2(1-z^2/3)dz}{1-\frac{\Lambda^2-1}{\Lambda^2}z^2}\\
&=\frac{{\Lambda}^{2} \left( 3\Lambda(2\Lambda^2-3){\rm arctanh} \left( {\frac {\sqrt {{\Lambda
}^{2}-1}}{\Lambda}} \right) +(8-5{\Lambda}^{2})\sqrt {{\Lambda}^{2}-1}
 \right)}{9\sqrt {{\Lambda}^{2}-1}({\Lambda}^{4}-2{\Lambda}^{2}+1)}\\
& = \frac{2}{3\pi}\log\Lambda-\frac{5}{9\pi}+\frac{2}{3\pi}\log2+O(\Lambda^{-2}\log\Lambda).
\end{align*}
Hence $ B_\Lambda^\zeta(0)= B_\Lambda^0(0)+O(\Lambda^{-2}\log\Lambda)$. Moreover, it can be seen that $B_\Lambda^\zeta(0)\leq 2/(3\pi)\log\Lambda$ when $\Lambda\geq4$.
The main result of this section is
\begin{prop}\label{prop_b_Lambda_zeta_particulier}Assume $\gH=L^2(\R^3,\C^4)$ and $\zeta(t)=t^2$. Let be $\alpha>0$ and $\Lambda>1$ such that $\alpha B_\Lambda^\zeta(0)<1$.
The function $b_\Lambda^\zeta$ satisfies
\begin{equation}
 \norm{b^\zeta_\Lambda}_{L^1(\R^3)}\leq \frac{\alpha B_\Lambda^\zeta(0)}{1-\alpha B_\Lambda^\zeta(0)}.
\label{estim_b_Lambda_particulier}
\end{equation}
\end{prop}
\begin{proof}
It follows from \eqref{formula_B_zeta_particulier} that
$$B_\Lambda^\zeta(k)  =\int_{0}^{1}\frac{8\Lambda^2dt}{\pi t^3}\int_0^{\sqrt{1-t^2}}\frac{(1-u^2)du}{1+3u^2}\frac{1}{\mu_1(t)^2+|k|^2}\times\frac{1}{\mu_2(t,u)^2+|k|^2}$$
where $\mu_1(t)=2/t$, and $\mu_2(t,u)=2\Lambda(1-1/\Lambda^2+1/t^2)^{1/2}(1+3u^2)^{-1/2}$.
The Fourier inverse of $(\mu^2+|k|^2)^{-1}$ is the Yukawa potential $e^{-\mu|x|}/(4\pi|x|)\geq0$. Therefore, the Fourier inverse $f_\Lambda=\cF^{-1}(B_\Lambda^\zeta)$ is nonnegative, so that
\begin{equation}
 \int f_\Lambda(x)\, dx=\int |f_\Lambda(x)|\, dx=(2\pi)^{-3/2}B_\Lambda^\zeta(0).
\label{integral_f_Lambda}
\end{equation}
In particular, the operator $T: g\in L^1(\R^3) \mapsto \cF^{-1}(\alpha B_\Lambda^\zeta\widehat{g})$ is bounded by $\norm{T}\leq \alpha B^\zeta_\Lambda(0)$. Hence, $1+T$ is invertible when $\alpha B^\zeta_\Lambda(0)<1$, and
$$\norm{(1+T)^{-1}}\leq \frac{1}{1-\alpha B^\zeta_\Lambda(0)}.$$
Proposition \ref{prop_b_Lambda_zeta_particulier} follows using that $b_\Lambda^\zeta=T(1+T)^{-1}$.
\end{proof}
\bibliographystyle{amsplain}

\end{document}